\documentclass[conference]{IEEEtran}
\IEEEoverridecommandlockouts
\usepackage{cite}
\usepackage{amsmath,amssymb,amsfonts}
\usepackage{graphicx}
\usepackage{textcomp}
\usepackage{xcolor}
\def\BibTeX{{\rm B\kern-.05em{\sc i\kern-.025em b}\kern-.08em
		T\kern-.1667em\lower.7ex\hbox{E}\kern-.125emX}}

\usepackage{etoolbox} % for toggle
\newtoggle{notarxiv}
\newcommand{\ifarxivelse}[2]{\iftoggle{notarxiv}{#1}{#2}}
\togglefalse{notarxiv}
\usepackage{amsmath,xcolor,graphicx,wrapfig,paralist}
\usepackage{algorithmicx,algorithm}
\usepackage[noend]{algpseudocode}
\usepackage{booktabs}
\usepackage{proof}
\usepackage{tikz}
\usepackage{csquotes}
\usepackage{graphics}
\usepackage{stmaryrd}
\usepackage{soul}
\usepackage{array}

\usepackage{ellipsis, mparhack, ragged2e} % bugfixing
\usepackage[l2tabu, orthodox]{nag} % avoid typical LaTeX errors

\usepackage{tabu,multirow}
\usepackage{xparse}
\usepackage{mathtools}
\usepackage{environ}
\usepackage{lscape}

\usepackage{url}

\usepackage[english]{babel}
\usepackage[pdftex]{hyperref}

\usepackage{subcaption}

\usepackage{algorithm}
\usepackage{algpseudocode}

\usepackage{marginfix}
\usepackage{xifthen}

\usepackage{listings}
\lstdefinestyle{customJ}{
  belowcaptionskip=1\baselineskip,
  breaklines=true,
  frame=L,
  xleftmargin=\parindent,
  language=Java,
  showstringspaces=false,
  basicstyle=\footnotesize\ttfamily,
  keywordstyle=\bfseries\color{green!40!black},
  commentstyle=\itshape\color{purple!40!black},
  identifierstyle=\color{blue},
  stringstyle=\color{orange},
}

\usepackage{xfrac}

\usepackage[official]{eurosym} % for Euro-Symbol

\usepackage{amsthm} % For thm environments; needs to be before cref
\usepackage{thmtools} % Not sure why, but this makes the restatable work also for lemmata and corollaries
\usepackage{thm-restate}
\usepackage[capitalise]{cleveref} % cref needs to be before redefining the environments

\newtheorem{theorem}{Theorem}[section]
\newtheorem{example}[theorem]{Example}
\newtheorem{lemma}[theorem]{Lemma}
\newtheorem{corollary}[theorem]{Corollary}
\newtheorem{definition}[theorem]{Definition}

\Crefname{equation}{Eq.}{Eqs.}
\Crefname{figure}{Fig.}{Figs.}
\Crefname{tabular}{Tab.}{Tabs.}
\Crefname{remark}{Rem.}{Rems.}
\Crefname{section}{Sec.}{Secs.}
\Crefname{subsection}{Sec.}{Secs.}
\Crefname{theorem}{Thm.}{Thms.}
\Crefname{example}{Ex.}{Exs.}
\Crefname{lemma}{Lem.}{Lems.}
\Crefname{corollary}{Cor.}{Cors.}
\Crefname{definition}{Def.}{Defs.}
\Crefname{appendix}{App.}{Apps.}
\Crefname{algorithm}{Alg.}{Algs.}

\crefname{equation}{Eq.}{Eqs.}
\crefname{figure}{Fig.}{Figs.}
\crefname{tabular}{Tab.}{Tabs.}
\crefname{remark}{Rem.}{Rems.}
\crefname{section}{Sec.}{Secs.}
\crefname{subsection}{Sec.}{Secs.}
\crefname{theorem}{Thm.}{Thms.}
\crefname{example}{Ex.}{Exs.}
\crefname{lemma}{Lem.}{Lems.}
\crefname{corollary}{Cor.}{Cors.}
\crefname{definition}{Def.}{Defs.}
\crefname{appendix}{App.}{Apps.}
\crefname{algorithm}{Alg.}{Algs.}
\newcommand{\para}[1]{\smallskip\noindent{\em #1}}

\newcommand{\intersectionSym}{\cap}
\newcommand{\intersectionBin}{\mathbin{\intersectionSym}}
\newcommand{\UnionSym}{\bigcup}

\newcommand{\intersection}{\intersectionBin}
\newcommand{\Union}{\UnionSym}

\DeclareGraphicsExtensions{.pdf, .png, .jpg}

\newtheorem{remark}{Remark}

\newcommand{\abs}[1]{\lvert #1 \rvert}

\newcommand{\Naturals}{\mathbb{N}}
\newcommand{\Reals}{\mathbb{R}}

\newcommand{\Rationals}{\mathbb{Q}}

\newcommand{\eqdef}{\vcentcolon=}

\newcommand{\PSPACE}{\mathbf{PSPACE}}

\newcommand{\PTIME}{\mathbf{PTIME}}

\newcommand{\MC}{\mathsf{M}}
\newcommand{\MDP}{\mathcal{M}}
\newcommand{\infinitepath}{\rho}
\newcommand{\Infinitepaths}{\mathsf{IPaths}}
\newcommand{\Finitepaths}{\mathsf{FPaths}}
\newcommand{\MEC}{\mathsf{MEC}}
\newcommand{\MECs}{\MEC}

\newcommand{\BSCCs}{\mathsf{BSCC}}

\newcommand{\QMDP}{{\widehat{\MDP}}}
\newcommand{\Qstates}{{\widehat{\states}}}
\newcommand{\Qinit}{{\widehat{\initstate}}}
\newcommand{\Qact}{{\widehat{\act}}}
\newcommand{\Qtrans}{{\widehat{\trans}}}
\newcommand{\Qstrat}{{\widehat{\strat}}}
\newcommand{\shat}{\widehat{s}}

\newcommand{\distribution}{d}
\NewDocumentCommand{\Distributions}{d()}{\IfNoValueTF{#1}{\mathcal{D}}{\mathcal{D}(#1)}}

\newcommand{\allstates}{\mathsf{S}}
\newcommand{\states}{\allstates}

\newcommand{\initstate}{s_{0}}
\newcommand{\act}{\mathsf{A}}

\newcommand{\trans}{\delta}

\newcommand{\targets}{\mathsf{F}}

\newcommand{\rew}{\mathsf{r}}

\newcommand{\meanpay}{\mathsf{mp}}

\newcommand{\sink}{\mathsf{z}}
\newcommand{\stay}{\mathsf{stay}}

\newcommand{\obj}{\Phi}
\newcommand{\objWR}{\Phi} 
\newcommand{\objMP}{\Phi^{MP}}
\newcommand{\objMO}{\mathcal{Q}}

\newcommand{\reach}{\lozenge}
\newcommand{\strat}{\sigma}
\newcommand{\Strats}{\Sigma}
\newcommand{\probability}{\mathbb{P}}

\newcommand{\val}{\mathsf{V}}

\newcommand{\achievable}{\mathcal{A}}
\newcommand{\pf}{\mathcal{P}}
\newcommand{\cp}{\mathsf{EP}}

\newcommand{\qee}{\hfill$\triangle$} % quod erat exemplandum % TODO del

\newcommand{\Expectation}{\mathbb{E}}

\newcommand{\CPTval}{\mathsf{CPT} \text{-} \val}
\newcommand{\cptfun}{\mathsf{cpt}}
\newcommand{\eu}{\mathsf{eu}}
\newcommand{\prosp}{\vec{x}}%\mathbf{x}}
\newcommand{\prospect}{\mathsf{prospect}}
\newcommand{\probabilities}{\vec{p}}
\newcommand{\probsappr}{\vec{p}_{\text{a}}}
\newcommand{\outcomevector}{\vec{o}}
\newcommand{\outcomes}{\mathsf{outc}}
\newcommand{\outcomeset}{\mathbb{O}}
\newcommand{\obtainset}{\mathsf{obtainset}}
\newcommand{\weight}{\mathsf{w}}
\newcommand{\pweight}{\weight^+}
\newcommand{\mweight}{\weight^-}
\newcommand{\weightP}{\pweight}
\newcommand{\weightM}{\mweight}

\newcommand{\gainrank}{\mathsf{grank}}
\newcommand{\lossrank}{\mathsf{lrank}}
\newcommand{\utility}{\mathsf{u}}
\newcommand{\util}{\utility}

\newcommand{\decweight}{\pi}
\newcommand{\probvector}{\probabilities}
\newcommand{\probvectorQ}{\vec{q}}
\newcommand{\lipCPT}{\mathsf{L}_{\cptfun}}
\newcommand{\lipWp}{\mathsf{L}_{\weightP}}
\newcommand{\lipWm}{\mathsf{L}_{\weightM}}

\newcommand*\myrefappendix[1]{%
		\ifbool{archiveversion}
		{\cref{#1}}
		{\cref{#1} in \cite{arxivVersion}}%
}

\makeatletter
\def\orcidID#1{\textsuperscript{\,\smash{\protect\raisebox{-1.25pt}{\href{http://orcid.org/#1}{\protect\includegraphics[scale=.8]{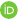}}}}}}
\makeatother
	
\begin{document}
	
	\title{Risk-aware Markov Decision Processes Using \\Cumulative Prospect Theory
		\thanks{This project has received funding from 
			the Fonds de la Recherche Scientifique - FNRS under grant No.\ T.0027.21, the Belgian National Lottery; %Thomas
			the ERC CoG 863818 (ForM-SMArt), the Austrian Science Fund (FWF) 10.55776/COE12; %Krish
			the DFG project 427755713, GOPro, and the DFG research training group GRK 2428 Continuous Verification of Cyber-Physical Systems (ConVeY); %Steffi
			the EU's Horizon 2020 research and innovation programmes under the Marie Sklodowska-Curie grant agreement No.\ 101034413 (IST-BRIDGE) and the ERC Starting Grant DEUCE (101077178). %Maxi
		}
	}

\author{
	\IEEEauthorblockN{
		Thomas Brihaye\IEEEauthorrefmark{1}\orcidID{0000-0001-5763-3130},
		Krishnendu Chatterjee\IEEEauthorrefmark{2}\orcidID{0000-0002-4561-241X},
		Stefanie Mohr\IEEEauthorrefmark{3}\orcidID{0000-0002-8630-3218},
		Maximilian Weininger\IEEEauthorrefmark{2}\IEEEauthorrefmark{3}\IEEEauthorrefmark{4}\orcidID{0000-0002-0163-2152}
	}
	\vspace{0.2cm}
	
	\IEEEauthorblockA{
		\begin{tabular}{cccc}
			~~~~\shortstack{\IEEEauthorrefmark{1}University of Mons\\
				Mons, Belgium \\
				Thomas.Brihaye@umons.ac.be
			} & 
			\shortstack{\IEEEauthorrefmark{2}Institute of Science and Technology Austria\\
				Klosterneuburg, Austria \\
				Krishnendu.Chatterjee@ist.ac.at
			} & 
			\shortstack{\IEEEauthorrefmark{3}Technical University of Munich\\
				Munich, Germany \\
				mohr@in.tum.de
			} &
		\end{tabular}
	}
\vspace{0.2cm}

	\IEEEauthorblockA{
			\shortstack{\IEEEauthorrefmark{4}Ruhr-University Bochum\\ Bochum, Germany\\
			Maximilian.Weininger@rub.de
			}
	}
}

	\maketitle

	\begin{abstract}
		Cumulative prospect theory (CPT) is the first theory for decision-making under uncertainty that combines full theoretical soundness and empirically realistic features~\cite[Page~2]{wakker2010prospect}.
		While CPT was originally considered in one-shot settings for risk-aware decision-making, we consider CPT in sequential decision-making.
		The most fundamental and well-studied models for sequential decision-making are Markov chains (MCs), and their generalization Markov decision processes (MDPs).
		The complexity theoretic study of MCs and MDPs with CPT is a fundamental problem that has not been addressed in the literature.
		
		Our contributions are as follows:
		First, we present an alternative viewpoint for the CPT-value of MCs and MDPs.
		This allows us to establish a connection with multi-objective reachability analysis and conclude the strategy complexity result that memoryless randomized strategies are necessary and sufficient for optimality.
		Second, based on this connection, we provide an algorithm for computing the CPT-value in MDPs with infinite-horizon objectives.
        We show that the problem is in EXPTIME and fixed-parameter tractable.
		Moreover, we provide a polynomial-time algorithm for the special case of MCs.
	\end{abstract}
	
	
	\section{Introduction}

\para{Risk-aware Decision Making.} 
A decision-maker wants to find the optimal choice in a situation with multiple possible decisions, but the outcomes of these decisions are affected by uncertainty. 
Such situations arise in countless applications, from strategic business decisions over financial investment to medical treatments, see e.g.~\cite[Chap. 1.1]{wakker2010prospect}.
\begin{example}[Motivating example]\label{ex:intro-motivate}
    Imagine you find a coupon for a bet.
    There are two bets available:
    You can use your coupon and choose to play a \emph{safe bet} with a 95\% chance of winning 20€.
    Alternatively, you can pay an additional 5€ and play a \emph{risky bet} where you can win 55€ with 51\%, 5€ with a 5\% chance, but you can also lose the bet and get nothing with a 44\% chance.
    Thus, taking into account the additional 5€ paid, the second bet can result in a loss of 5€, no change, or winning 50€.
    Most people tend to choose the safe bet.
    This has been claimed to be \enquote{irrational} because the expectation of the safe bet is only 19€, while the expectation of the risky bet is higher~at~23.3€.
    We discuss a sequential version of this example (finding two coupons) in \cref{ex:2-running}.
    \qee
\end{example}

\para{Cumulative Prospect Theory over Expected Utility.}
Expected utility theory cannot explain human behaviour in examples like this. 
Thus, there has been an effort to replace it with another decision theory that faithfully represents human decision-making in order to predict or prescribe decisions that align with human preferences.
We refer to~\cite{fishburn1988expected} for a history of expected utility theory and the criticism it has received and to~\cite[P. 2]{wakker2010prospect} for an overview of the developments that led to the inception of \emph{cumulative prospect theory} (CPT), introduced in~\cite{CPT92}.
Unlike expected utility theory, CPT correctly models the decision of preferring the safe bet in the motivating example, as it accounts for factors in human decision-making like loss-aversion (losing money is more hurtful than gaining an equal amount) and re-weighting of probabilities (the small 5\% chance for getting nothing is increased in human perception, while the larger 95\% chance is decreased).

\para{Desired Properties of CPT.}
CPT combines theoretical soundness with empirical realism:
It is sound insofar as it satisfies two important axioms, called gain-loss consistency and sign-comonotonic tradeoff consistency, whereas other theories of decision under risk do not~\cite{CPTaxioms}.
It is empirically realistic, as it has been confirmed by several studies, see~\cite[Chap. 9.5]{wakker2010prospect} for an overview. 
In 2002, Daniel Kahneman received the Nobel Memorial Prize in Economic Sciences for his work on prospect theory.
In 2010, CPT was \enquote{still the only theory that can deliver the full spectrum of what is required for decision under uncertainty}~\cite[P. 2]{wakker2010prospect}.

\para{CPT in Sequential Decision Making.} 
CPT was originally considered in one-shot settings, but is also relevant in the context of sequential decision-making. 
For this, the most fundamental and well-studied models are Markov chains (MCs, e.g.~\cite{kulkarnimodeling-second-edition}), and their generalization Markov decision processes (MDPs, e.g.\cite{puterman}). 
Adapting the motivating example, MDPs can model not only a single betting decision but when we find two coupons for bets and we can bet twice consecutively.
An optimal strategy is a risk-aware plan for both bets. 
Usually, the value of an MDP is calculated as the expectation over the rewards of all paths. 
Replacing this expectation with the CPT-function allows risk-aware sequential decision-making in problems too large for humans to comprehend.
The computational study of MDPs with CPT is a fundamental algorithmic problem in risk-aware sequential decision making which has not been addressed in the literature.

\para{Our contribution.}
We focus on the algorithmic analysis of CPT in MCs and MDPs with infinite-horizon objectives, namely the basic and well-studied weighted reachability and long-run average reward. 
We provide the first deterministic algorithms and the first computational and strategy complexity results.
\begin{compactitem}
    \item First, we present a new, more intuitive definition of the CPT-value in sequential decision-making, which allows us to establish a connection to multi-objective reachability queries.
    We utilize this connection to apply results from multi-objective model checking.
    This way, we directly obtain the strategy complexity result that memoryless randomized strategies are necessary and sufficient for achieving the optimal CPT-value.
    \item Second, we provide algorithms for approximating the CPT-value of an MC or MDP to arbitrary precision $\varepsilon>0$. 
    (We explain in \Cref{rem:2-approx} why computing the exact CPT-value is unreasonable and infeasible.)
    Based on this, we also establish upper bounds on the computational complexity of the CPT-value approximation problem:
    Given an MDP, weighted reachability objective, precision $\varepsilon$, and threshold $v$, decide whether the CPT-value of the MDP is greater than $v+\varepsilon$ or smaller than $v-\varepsilon$.
    This CPT-value approximation is fixed-parameter tractable.
\end{compactitem}

\para{Technical contributions.}
One of our key contributions is a new and intuitive definition of the CPT-value in the sequential setting (\cref{sec:3-new-cpt-value-title}). 
The new definition immediately leads to our polynomial-time algorithm for Markov
chains (\cref{sec:4-mc-algo}). 
For MDPs, the new definition allows us to establish a connection to multi-objective reachability (\cref{sec:3-relation-cpt-multi-objective}). 
For the algorithmic results, we exploit this connection in order to reduce the problem of computing the CPT-value to a non-convex optimization problem (\cref{sec:4-mdp-algo}). 
The complexity of non-convex optimization is usually not discussed and only described as exponential in the number of variables, see e.g.~\cite[Chap. 3]{matthiesen2019efficient}.
We provide an algorithm with explicit runtime analysis.
Finally, we highlight that our theory is robust in the sense that we can generalize it to variants of CPT or to other objectives like mean payoff (\cref{sec:5-title}).

\subsection*{Related Work.}
\para{Risk-aware MDPs.}
Recently, there has been a lot of work on risk-aware MDPs.
Often, expectation is not replaced but rather complemented with something by considering a conjunction of objectives. This includes combining expectation with conditional value-at-risk~\cite{KM18-cvar,DBLP:conf/aaai/Meggendorfer22}, percentiles~\cite{DBLP:journals/fmsd/RandourRS17}, quantiles~\cite{DBLP:conf/fossacs/UmmelsB13,DBLP:journals/mor/JiangP18}, or variance~\cite{DBLP:conf/icml/MannorT11,DBLP:journals/jcss/BrazdilCFK17}.

In contrast, CPT is a \enquote{monolithic} theory of decision, i.e.\ it computes a single value for every choice, such that the largest value is the one that a human would prefer.
A similar approach is used in ~\cite{erisk,DBLP:journals/iandc/BaierCMP24}, replacing expectation with entropic risk, or~\cite{DBLP:conf/icalp/PiribauerSB22,DBLP:journals/eor/MaMX23,DBLP:conf/concur/BaierPS24}, using variance-penalized expected payoff.

In~\cite{DBLP:conf/cdc/Ramasubramanian21,DBLP:journals/ftml/A022}, the CPT-value in MDPs is approximated using model-free reinforcement learning algorithms.
These randomized algorithms only converge in the limit and have no stopping criterion indicating how close the current approximation is to the result; in contrast, our solution is deterministic and terminates with a required precision after finitely many steps.
In~\cite{lin2013dynamic,lin2018probabilistically,TopcuConference}, the authors consider a \enquote{nested} formulation of the CPT-value, i.e.\  the CPT-function is applied for every transition that a path in the MDP takes.
The convergence of their algorithms is proven only under additional restrictive assumptions, namely that the reward measure is b-bounded~\cite[Theorem~3.1]{lin2013dynamic} (also used in \cite{TopcuConference}) or that the MDP is uniformly transient~\cite[Theorem~10]{lin2018probabilistically}.
Additionally, they differ from our work in that the considered objectives are finite-horizon~\cite{lin2013dynamic,TopcuConference} or discounted and transient~\cite{lin2018probabilistically}, whereas we consider undiscounted infinite-horizon objectives.

\para{Comparison of Theories of Risk.}
CPT is the dominant theory of risk in psychology, and it generalizes classical expected utility, as well as the value-at-risk~\cite[Exerc. 6.4.4]{wakker2010prospect} in the sense that for a certain choice of internal parameters, CPT degenerates to these functions.
It is based on empirical studies~\cite{CPT92} and has been mathematically axiomatized, in particular satisfying the properties of gain-loss consistency and sign-comonotonic tradeoff consistency~\cite{CPTaxioms}.
The conditional value-at-risk and entropic risk have also been mathematically axiomatized in~\cite{DBLP:journals/mansci/WangZ21} and~\cite{DBLP:journals/fs/FollmerS02}, respectively, and are often employed in finance. Entropic risk has been shown to be consistent with human choice~\cite{DBLP:journals/eor/BrandtnerKR18} in portfolio selection; however, as it developed from an axiomatization (in contrast to CPT, where the axiomatization was derived from the empirical studies), there are fewer empirical studies and an ongoing debate about which variant of entropic risk is preferable~\cite{follmer2011entropic,DBLP:journals/eor/BrandtnerKR18,fischer2018discussion}.
The variance-penalized expected payoff is also a very common approach, and Harry Markowitz received the Nobel Memorial Prize in Economic Sciences for it~\cite{markowitz1991foundations}; however, the theoretical foundations have been subject to a lot of discussion, see the summary in~\cite{johnstone2013mean}, and we are not aware of an axiomatic characterization. 
Overall, there are arguments for investigating MDPs under these theories of risk, and we complement the state-of-the-art by providing results for CPT, for which so far there is no principled analysis of algorithms and complexity.

\para{Algorithms for CPT in other settings.} 
There exist algorithms converging in the limit to the CPT-value of a discounted deterministic two-player Markov game~\cite{DBLP:conf/aaai/TianST21}, or to approximate Nash-equilibria in multi-player stochastic games with incomplete information~\cite{DBLP:journals/tac/EtesamiSMP18,danis2023multi}. 
None of these papers provide results on algorithmic or strategy complexity.
	\section{Preliminaries}\label{sec:2-title}

These preliminaries are quite extensive, as they recall material from several research fields. 
Thus, some definitions not necessary for understanding the key insights of the paper as well as more extensive examples appear only in \ifarxivelse{\cite[App. A]{techreport}}{\cref{app:2-extended-prelims}}.
The preliminaries are organized as follows:
Firstly, \Cref{sec:2-main-cpt} provides the definition of cumulative prospect theory (CPT) and the CPT-function $\cptfun$; it additionally recalls the expected utility function $\eu$, because comparing the two functions is instructive.
Secondly, \Cref{sec:2-markov-subsection} briefly introduces Markov systems, i.e.\ Markov decision processes (MDPs) and their semantics, in particular defining the weighted reachability objective we consider throughout most of the paper.
\Cref{sec:2-problem-statement} combines the previous subsections to define the CPT-value of an MDP (and its expected value, again for accessibility and comparison) and provides the formal problem statement.
Finally, \Cref{sec:2-further} recalls further notions necessary for our technical reasoning: strongly connected components, end components and multi-objective reachability queries. 

\para{Conventions.}
We denote vectors using an arrow above a letter, e.g., $\vec{o}$.
For a $k$-vector $\vec{o}$ and for $1\leq i\leq k$, we write $o_i$ for the $i$-th element of this vector.
Vectors are compared point-wise, i.e.\ for ${\sim} \in \{<,>,\leq,\geq,=\}$, we have $\vec{x} \sim \vec{y}$ if and only if for all $1\leq i\leq k$ it holds that $x_i \sim y_i$.

$\Distributions(X)$ denotes the set of all \emph{probability distributions} over a finite set $X$, i.e.\ mappings $\distribution : X \to [0, 1]$ such that $\sum_{x \in X} \distribution(x) = 1$.
For a fixed ordering of elements in $X$, a distribution can be interpreted as a $k$-vector of probabilities, i.e.\ an element of $[0,1]^k$.

\subsection{Cumulative Prospect Theory}\label{sec:2-main-cpt}

This section recalls the basics of cumulative prospect theory (CPT) based on the textbook~\cite{wakker2010prospect}.

\subsubsection{Outcomes and Prospects}

Let $\outcomeset \subseteq \Reals$ be a finite set of \emph{outcomes}, i.e.\ of possible gains and losses that one can obtain. We assume without loss of generality that we always have $0 \in \outcomeset$.
We use $\outcomevector$ to denote the vector containing all outcomes, and assume w.l.o.g. that they are ordered increasingly (this is necessary for evaluating the CPT-function).
Throughout the paper, we denote by $k \eqdef \abs{\outcomeset}$ the number of outcomes.

A \emph{prospect} (also called lottery) is a tuple $(\outcomevector, \probabilities)$, where $\outcomevector$ is a vector of outcomes and $\probabilities \in [0,1]^k$ is a vector of probabilities.
Intuitively, a prospect is a probability distribution over $\outcomeset$.
We denote by $\prosp = (\outcomevector, \probabilities)$ a prospect consisting of $\outcomevector$, the ordered $k$-vector of outcomes, and $\probabilities$, the associated $k$-vector of probabilities.
In our examples, we write prospects in a list form, e.g.\ $\prosp \eqdef [o_1: p_1, \ldots, o_k: p_k]$, where $o_1,\ldots,o_k \in \outcomeset$ are the outcomes and $p_1,\ldots,p_k$ are their respective probabilities.

\begin{example}
	\label{ex:prospects}
	In \cref{ex:intro-motivate}, the set of outcomes is     $\{-5,0,20,50\}$
    and we can choose between two prospects 
    $[0:0.05, 20:0.95]$ and $[-5:0.44, 0: 0.05, 50: 0.51]$.
    \qee
\end{example}

\subsubsection{Evaluating Prospects Using Expected Utility Theory}\label{sec:2-prelim-eu}

When deciding which prospect to prefer, humans evaluate them intuitively; CPT is one model for this evaluation.
For better accessibility of the CPT-function, we first recall the well-known and simple definition of expected utility.
It is a function $\eu: \outcomeset^k \times [0,1]^k \to \Reals$ that takes a prospect and computes its expected value while also re-weighting every outcome with a utility function $u:\outcomeset \to \Reals$.
This utility function already allows capturing some parts of human behaviour, for example, the fact that the perceived gain between no money and one million Euros is larger than between the first and second million. 
Commonly, utility functions have an \enquote{S-shape}, as we exemplify in \Cref{fig:2-functions}. 
Formally, $\eu$ is defined on a prospect $\prosp = (\outcomevector,\probabilities)$ as follows:
\begin{equation}\label{eq:2-eu-app}
	\eu(\prosp)=\eu((\outcomevector,\probabilities)) = \sum_{i=1}^{k} \util(o_i) \cdot p_i
\end{equation}

\subsubsection{Evaluating Prospects Using Cumulative Prospect Theory}\label{sec:2-prelim-cpt-eval}
The CPT-function is more complicated than expected utility because it combines many particularities of human preference, see \cite{CPT92,wakker2010prospect}.
It uses both a utility function $\util$ and two probability weighting functions $\weightP$ and $\weightM$, which all 
%All of these three functions 
are increasing (f is \emph{increasing} if for $x\geq x'$, $f(x)\geq f(x')$), continuous and exhibit an \enquote{S-shape} similar to that in \Cref{fig:2-functions}.
We depict and explain common choices in \ifarxivelse{\cite[App. A-A]{techreport}}{\Cref{app:2-u-and-w}}.
Here, we only highlight that the functions usually are non-convex and employ roots, and thus applying them can result in irrational numbers.
\begin{figure}
	\centering
		\centering
        \includegraphics[width=0.4\linewidth]{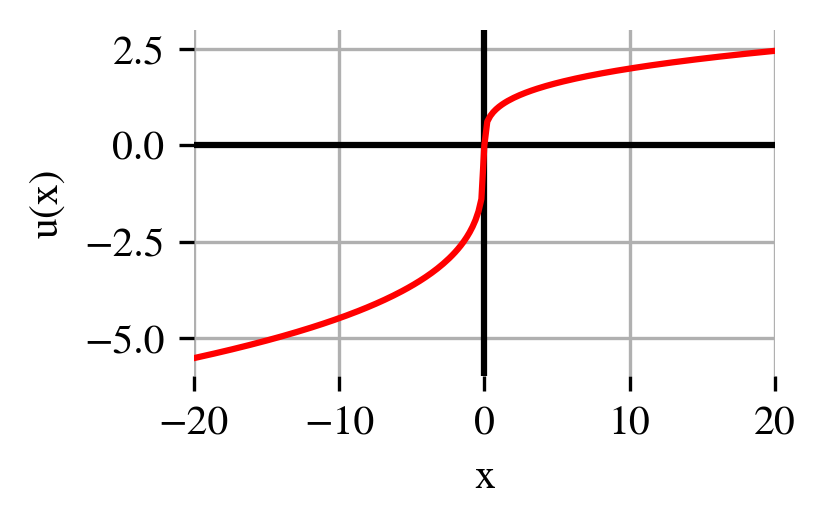}
	\caption{A depiction of a common utility function $\util$.
	}
	\label{fig:2-functions}
\end{figure}
The CPT-function $\cptfun: \outcomeset^k \times [0,1]^k \to \Reals$ maps a prospect $\prosp = (\outcomevector,\probabilities)$ to a real number as follows:
\begin{equation}\label{eq:2-cpt}
	\cptfun(\prosp) = \sum_{i=1}^{k}  \util(o_i) \cdot \decweight(\prosp, i),
\end{equation}
where $\decweight: (\outcomeset^k \times [0,1]^k) \times \Naturals\to \Reals$ is the \emph{decision weight} of the $i$-th outcome.
For classic expected utility theory, we have $\decweight(\prosp, i) \eqdef p_i$.
For CPT, $\decweight$ depends on a ranking of the outcomes and the probability weighting functions $\weightP$ and $\weightM$, and we provide the definition in \ifarxivelse{\cite[App. A-B]{techreport}}{\cref{app:2-CPT}}.
We highlight that for understanding the paper (except for some technical details in the proofs), it is sufficient to consider the function $\cptfun$ as a \textbf{black-box that maps prospects to reals}.

\subsection{Markov Systems and Their Semantics}\label{sec:2-markov-subsection}

\subsubsection{Markov Systems}

A \emph{Markov chain (MC)} (e.g.\ \cite{kulkarnimodeling-second-edition}) is a tuple $\MC = (\states,\initstate,\trans)$,
where $\states$ is a countable, non-empty set of \emph{states};
$\initstate\in\states$ is the \emph{initial state}; 
and $\trans \colon \states \to \Distributions(\states)$ is the \emph{transition function} that yields a probability distribution over successor states for each state.
$\MC$ is called finite if $\states$ is finite.

A (finite-state) \emph{Markov decision process (MDP)} (e.g.\ \cite{puterman}) is a tuple $\MDP = (\states, \initstate, \act, \trans)$, 
where $\states$ is a \emph{finite} set of states; 
$\initstate \in \states$ is the initial state; 
$\act$ denotes a finite set of \emph{actions} and we overload $\act$ to also assign to each state $s$ a non-empty set of \emph{available actions} $\act(s)$; 
and $\trans \colon \states \times \act \to \Distributions(\states)$ is the (partial) \emph{transition function} that for every state $s$ and available action $a \in \act(s)$ yields a distribution over successor states.

An MC is an MDP where every state has exactly one action, see~\cite[Cor. 3.1]{de1997formal}.
We use this in the following to define concepts only for MDPs, even though they also apply to MCs.
\cref{fig:motivating-example-MDP} shows our running example (\cref{ex:intro-motivate}) as an MDP.

\subsubsection{Semantics: Paths and Strategies}\label{app:2-mdp-semantics}

The \emph{semantics} of MDPs are defined as usual: The non-deterministic choices (which actions to choose) are resolved by means of a strategy (also called policy or scheduler).
This induces an MC with the respective probability space over infinite paths; we briefly recall this approach and refer to~\cite[Chap. 10]{BK08} for a more extensive description.
We use the following standard notation: For a set $S$, we denote by $S^*$ and $S^\omega$ the set of finite and infinite sequences comprised of elements of $S$, respectively.

In every state $s$, we choose an available action $a \in \act(s)$ and advance to a successor state $s'$ according to the probability distribution given by $\trans(s, a)$.
Starting in a state $s_0$ and repeating this process indefinitely yields an infinite path $\infinitepath = s_0 a_0 s_1 a_1 \dots \in (\states \times \act)^\omega$ such that (i) for every $i \in \Naturals_0$ we have $a_i \in \act(s_i)$ and (ii) $s_{i+1} \in \{s' \in \allstates \mid \trans(s_i, a_i)(s') > 0\}$.
A finite path is a prefix of an infinite path ending in a state, i.e.\ an element of $(\states \times \act)^\ast \times \states$ satisfying the conditions (i) and (ii).
$\Finitepaths_\MDP$ and $\Infinitepaths_\MDP$ denote the set of all such finite and infinite paths in a given MDP $\MDP$.
Further, we write $\infinitepath_i$ to denote the $i$-th state $s_i$ in a path.

A \emph{strategy} (also policy or scheduler) describes a way to choose actions.
In full generality, a strategy is a function $\strat \colon \Finitepaths_\MDP \to \Distributions(\act)$.
It samples an action from a distribution that is chosen depending on the history (the finite path observed so far).
We differentiate \emph{history-dependent} strategies from \emph{memoryless} strategies, whose choice depends only on the last state of the finite path.
We differentiate strategies that are \emph{randomized} from those where the distribution over actions is Dirac; these choose a single action surely and are called \emph{deterministic}.
We denote the set of all strategies by $\Strats$.

Resolving all nondeterministic choices by fixing a strategy $\strat$ yields a (countably infinite) Markov chain~\cite[Def. 10.92]{BK08}, denoted by $\MDP^{\strat}$.
This MC  induces a unique probability distribution $\probability_{\MDP}^{\strat}$ over the set of all infinite paths $\Infinitepaths_\MDP$ \cite[Chap.~10.1]{BK08} (where paths starting in $\initstate$ have measure 1).
Referring to the probability measure of an MC $\MC$, we omit the superscript and write $\probability_{\MC}$.

\subsubsection{Objectives}\label{sec:2-objective}

An \emph{objective} $\obj : \Infinitepaths_\MDP \to \Reals$ formalizes the \enquote{goal} of the player by assigning a value to each path.
In this paper, we mainly focus on \emph{weighted reachability (WR)} (e.g.,  \cite{everett1957recursive,KM18-cvar,AM09}, also called recursive payoff function or terminal reward).\footnote{Weighted reachability sometimes, e.g., in \cite{DBLP:conf/rp/BrihayeG23}, refers to minimizing the cost before reaching a target. Our techniques do not apply 
to this objective, as we discuss in 
\ifarxivelse{\cite[App. D-C]{techreport}}{\Cref{app:5-num-outcomes}}.\label{fn:2-weighted-reach}.}
Intuitively, this objective assigns rewards to some selected target states.
The value of a path is the reward of the first visited target state.
Formally, let $\targets\subseteq\states$ be a set of target states and $\rew \colon \targets \to \Rationals$ a reward function assigning rewards to them. 
We have $\objWR(\infinitepath) = \rew(\min_{i\in\Naturals} \{\infinitepath_i \mid \infinitepath_i \in \targets\})$ if such an $i$ exists and 0 otherwise. 
One can choose an arbitrary value instead of 0 for paths not reaching the target, see \Cref{sec:5-title}.

We assume w.l.o.g. that all target states are absorbing, i.e.\ for all $s\in\targets$ we have $\trans(s,a)(s)=1$ for all $a\in\act(s)$; this is justified because after reaching a target, the value of a path is fixed.
Further, we assume that $\rew(s) \neq 0$ for all $s\in\targets$, as the reward of 0 is also obtained by not reaching a target; this simplifies the notation of the concepts introduced in \cref{sec:3-new-cpt-value-title}.

In the following, we overload $\obj$ to denote both the tuple $(\targets, \rew)$ describing the weighted reachability objective as well as the induced function that maps infinite paths to reals.

Our results carry over to the mean payoff objective (e.g.~\cite{puterman}, also called long-run average reward), as we explain in \Cref{sec:5-title}.
In order to not detract from the key message, we delay the formal definition of mean payoff to that section.

\begin{figure}
	\centering
     \begin{tikzpicture}[xscale=1,yscale=0.7]
    \node[draw,circle, minimum size=0.5cm] (s0) at (0,0){$s_0$};

    \node[draw,circle,fill=black,inner sep=0pt,minimum size=5pt] (a1) at (-1.5,0){};
   
    \node[draw,circle,fill=black,inner sep=0pt,minimum size=5pt] (a2) at (1.7,0){};

    \node[draw,circle, minimum size=0.5cm] (s1) at (-3,-1.2){$s_1$};
    \node[draw,circle, minimum size=0.5cm] (s2) at (-0.5,-1.2){$s_2$};
    \node[draw,circle, minimum size=0.5cm] (s3) at (1,-1.2){$s_3$};
    \node[draw,circle, minimum size=0.5cm] (s4) at (3,-1.2){$s_4$};

    \node (s1c) at (-3,-2){$20$};
    \node (s2c) at (-0.5,-2){$0$};
    \node (s3c) at (1,-2){$-5$};
    \node (s4c) at (3,-2){$50$};

    \draw[->,thick] (-.75,0.75) -- (s0);
    \draw[->,thick] (s0) -- node[fill=white, inner sep=1pt] {$a_1$} (a1);
    \draw[->,thick] (s0) -- node[fill=white, inner sep=1pt] {$a_2$} (a2);

    \draw[->,thick] (a1) -- node[fill=white, inner sep=1pt] {$.95$} (s1);
    \draw[->,thick] (a1) -- node[fill=white, inner sep=1pt] {$.05$} (s2);

    \draw[->,thick] (a2) -- node[fill=white, inner sep=1pt] {$.05$} (s2);
    \draw[->,thick] (a2) -- node[fill=white, inner sep=1pt] {$.44$} (s3);
    \draw[->,thick] (a2) -- node[fill=white, inner sep=1pt] {$.51$} (s4);

\end{tikzpicture}
	\caption{\cref{ex:intro-motivate} as MDP. 
		We omit the self looping actions on states $s_1, \ldots, s_4$.}
	\label{fig:motivating-example-MDP}
\end{figure}
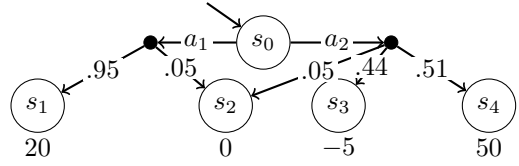

\begin{example}[Extended version in \ifarxivelse{\cite[App. A-C]{techreport}}{\cref{app:2-exs}}]\label{ex:2-running}
	We model \cref{ex:intro-motivate} as an MDP $\MDP=(\states,\initstate,\act,\trans)$ (shown in \Cref{fig:motivating-example-MDP}).
	We have $\states=\{s_0,s_1,s_2,s_3,s_4\}$, and $\act(s_0)=\{a_1,a_2\}$, with action $a_1$ representing the safe bet with a  95\% chance of 20€, $\trans(s_0,a_1)=(s_1\rightarrow 0.95, s_2\rightarrow 0.05)$, and $a_2$, the risky bet, modeled as $\trans(s_0,a_2)=(s_2\rightarrow 0.05, s_3\rightarrow 0.44, s_4\rightarrow 0.51)$.
	The set of target states is $\targets=\{s_1,s_3,s_4\}$ with rewards
    $r(s_1)=20$, $r(s_3)=-5$, $r(s_4)=50$.
	$s_2$ yields a reward of 0, because there is no path to a target state.
	
	The actions correspond to the following prospects: 
    $\prosp_1=[0:0.05,\;20:0.95]$ and $\prosp_2=[-5:0.44,\;0:0.05,\;50:0.51]$.
	Using the standard choices for $\util$ and $\decweight$ (see \ifarxivelse{\cite[App. A-A]{techreport}}{\cref{app:2-u-and-w}}), we compute 
    $\cptfun(\prosp_1) \approx 11.07$ and $\cptfun(\prosp_2) \approx 10.19$.
    Using the expected utility function $\eu$ with $\util(x)=x$, we get 
    $\eu(\prosp_1)=19$ and $\eu(\prosp_2)=23.3$.
	Thus, expected utility theory recommends choosing $a_2$ while CPT predicts that most humans prefer $a_1$. 

    For a truly sequential example, we extend this by finding two coupons for bets.
    We depict in \ifarxivelse{\cite[Fig. 7]{techreport}}{\Cref{fig:app-motivating-example-MDP-sequential} in \cref{app:2-exs}} the MDP, where we have the same choice between $a_1$ and $a_2$ a second time in each of the states $s_1,s_2,s_3,s_4$.
    For each combination of decisions between safe and risky, we get a prospect.
    Evaluating both expected utility $\eu$ and CPT $\cptfun$, we see that expected utility chooses the risky bet twice, whereas CPT recommends the safe and risky bet each once.
	\qee
\end{example}

\subsection{CPT-value and Problem Statement}\label{sec:2-problem-statement}

\subsubsection{Expected Value and CPT-value}\label{sec:2-expval-and-cptval}
The definition of the value of an MDP under CPT is quite complicated.
For easier accessibility, we first recall the standard way, i.e.\ the value using expected utility theory.
Observe that $\obj(\infinitepath)$ is a random variable; after choosing a strategy $\strat\in\Strats$, it is distributed according to the probability measure $\probability_{\MDP}^{\strat}$.
Then, the value under expected utility theory is the expectation of this random variable, where the expectation over a countable set $X$ under a probability measure $\probability$ is defined as usual: $\sum_{x\in X} x \cdot \probability[x]$.
Denoting the expectation under the probability measure as $\Expectation_{\MDP}^{\strat}$, we define the value as follows:
\begin{equation}
	\val(\MDP,\obj) \eqdef \sup_{\strat\in\Strats} \Expectation_{\MDP}^{\strat}[\util(\obj(\infinitepath))]\label{eq:2-eu-val}
\end{equation}
Since one can rescale the rewards and use $\util(x)=x$, the utility function is often omitted.

For the CPT-value, we replace the expectation with the CPT-function.
Here, we only show the definition that mimics the style of other papers considering MDPs with CPT, e.g.~\cite[Chap. 3.7]{DBLP:journals/ftml/A022}.
However, we delegate the extensive explanation to \ifarxivelse{\cite[App. A-D]{techreport}}{\cref{app:2-cptval}}, as our first contribution is to replace it with a more intuitive and easier-to-read definition in \cref{sec:3-new-cpt-value}.
We define positive and negative outcomes as $\outcomes^+(\objWR) \eqdef \{ o\in \outcomes(\objWR)\mid o>0\}$, and analogously $\outcomes^-(\objWR)$. 
\vspace{-1em}
{\par\nopagebreak\small%
\begin{align}
	\scriptsize
	\label{eq:2-cpt-val}
	&\CPTval(\MDP,\obj) \eqdef \sup_{\strat\in\Strats} \bigg(\\
	&
	\sum_{o \in \outcomes^+(\obj)} \util(o) \cdot \left( \weightP\left( \probability_{\MDP}^{\strat}(\obj(\infinitepath) \geq o) \right) - \weightP\left( \probability_{\MDP}^{\strat}(\obj(\infinitepath) > o) \right)\right) + \nonumber \\
	&\sum_{o \in \outcomes^-(\obj)} \util(o) \cdot \left(\weightM\left( \probability_{\MDP}^{\strat}(\obj(\infinitepath) \leq o) \right) - \weightM\left( \probability_{\MDP}^{\strat}(\obj(\infinitepath) < o) \right)\right) \nonumber
	\bigg)
\end{align}
}%

\subsubsection{Problem statement}\label{sec:2-problem-statement-formal}

We formalize our problem statement as follows:
\medskip

\noindent\framebox[\linewidth]{\parbox{0.95\linewidth}{
		\underline{\textbf{CPT-value approximation problem}}
		
		\textbf{Given} an MDP or MC $\MDP$, a weighted reachability objective $\obj$, a precision $\varepsilon>0$ and a threshold $v$.
		
		\textbf{Return} $
		\begin{cases}
			\mathsf{true} &\mbox{if } \CPTval(\MDP,\obj) > v+\varepsilon\\
			\mathsf{false} &\mbox{if } \CPTval(\MDP,\obj) < v-\varepsilon\\
			\text{arbitrary} &\mbox{otherwise}
		\end{cases}
		$
}}

\begin{restatable}[Exact computation]{remark}{remarkExact}
	\label{rem:2-approx}
	There are two reasons why we are only interested in approximating the CPT-value, namely that an exact computation is unreasonable and infeasible.
	It is unreasonable because the parameters of the CPT-function, e.g.\ the utility function $\util$ and the decision weight $\decweight$, are only empirically estimated by analyzing the behaviour of many people~\cite[Chap. 9.5]{wakker2010prospect}.
	Thus, naturally, they are only approximations of the \enquote{true} functions for a particular person.
	It is infeasible because evaluating the CPT-function on a prospect can result in irrational numbers, as, for example, the utility function uses roots, see \ifarxivelse{\cite[App. A-A]{techreport}}{\Cref{app:2-u-and-w}}.
\end{restatable}

\subsection{Technical Notions Required for the Solution}\label{sec:2-further}

For the analysis of MDPs, the following two graph theoretical concepts are of great importance: \emph{strongly connected components} and \emph{maximal end components}.

\subsubsection{Strongly Connected Components}\label{sec:2-graph-components}

A non-empty set of states $C \subseteq \states$ in an MDP is \emph{strongly connected} if for every pair $s, s' \in C$ there is a non-empty finite path from $s$ to $s'$.
Such a set $C$ is a \emph{strongly connected component} (SCC) if it is inclusion maximal, i.e.\ there exists no strongly connected $C'$ with $C \subsetneq C'$.
SCCs are disjoint, so each state belongs to at most one SCC.
An SCC is \emph{bottom} (BSCC) if additionally no path leads out of it, i.e.\ for all $s \in C$, $a\in\act(s)$, $s' \in S \setminus C$ we have $\trans(s, a, s') = 0$.
The set of BSCCs of an MDP $\MDP$, denoted $\BSCCs(\MDP)$, can be determined in linear time~\cite{DBLP:journals/siamcomp/Tarjan72}.

A state might not be in a BSCC in an MDP, but in a BSCC in the induced MC under some strategy, see \cref{ex:3-stopping}.

\subsubsection{Maximal End Components}\label{sec:2-end-components}

A pair $(T, B)$, where $\emptyset \neq T \subseteq \states$ and $\emptyset \neq B \subseteq \Union_{s \in T} \act(s)$, is an \emph{end component (EC)} of an MDP $\MDP$ if (i)~for all $s \in T, a \in B \intersection \act(s)$ we have $\{s' \mid \trans(s, a,s')>0\} \subseteq T$, and (ii)~for all $s, s' \in T$ there is a finite path $\infinitepath = s a_0 \dots a_n s' \in (T \times B)^* \times T$, i.e.\ the path stays inside $T$ and only uses actions in $B$.
Intuitively, an EC describes a set of states where for some particular strategy all possible paths remain inside these states (i.e.\ under this strategy, the EC is a BSCC in the induced MC).
An EC $(T, B)$ is a \emph{maximal end component (MEC)} if there is no other EC $(T', B')$ such that $T \subseteq T'$ and $B \subseteq B'$.
The set of MECs in a subset $\states' \subseteq \states$, denoted by $\MECs_\MDP(\states')$, can be computed in polynomial time~\cite{DBLP:journals/jacm/CourcoubetisY95}.

\subsubsection{Multi-Objective Reachability Queries}\label{sec:2-MO-query}
For our technical reasoning, we additionally require \emph{multi-objective} (non-weighted) reachability queries, see e.g.~\cite{EKVY08}.
These are given as an n-tuple of target sets $\objMO = (\targets_1,\ldots,\targets_n)$.
Intuitively, we want to compute the optimal probability to reach all of these target sets; however, there is no single optimal probability but multiple incomparable points as they trade off reaching one goal with reaching another.

\begin{figure}
	\centering
	\begin{subfigure}{0.49\linewidth}
		\centering
		 \begin{tikzpicture}[scale=1]
    \node[draw,circle, minimum size=0.5cm] (s0) at (0,0){$s_0$};

    \node[draw,circle,fill=black,inner sep=0pt,minimum size=5pt] (b) at (.75,1){};

    \node[draw,circle,fill=black,inner sep=0pt,minimum size=5pt] (a) at (1,0){};
   
    \node[draw,circle,fill=black,inner sep=0pt,minimum size=5pt] (c) at (.75,-1){};

    \node[draw,circle, minimum size=0.5cm,fill=blue!50!white] (s1) at (2,1.3){$s_1$};
    \node[draw,circle, minimum size=0.5cm] (s2) at (2.5,0){$s_2$};
    \node[draw,circle, minimum size=0.5cm,fill=red!30!white] (s3) at (2,-1.3){$s_3$};

    \draw[->,thick] (-.75,0) -- (s0);
    \draw[->,thick] (s0) -- node[fill=white, inner sep=0pt] {$a$} (a);
    \draw[->,thick] (s0) -- node[fill=white, inner sep=0pt] {$b$} (b);
    \draw[->,thick] (s0) -- node[fill=white, inner sep=0pt] {$c$} (c);

    \draw[->,thick] (b) -- node[fill=white, inner sep=0pt] {$.8$} (s1);
    \draw[->,thick] (b) -- node[pos=.3,fill=white, inner sep=0pt] {$.2$} (s2);

    \draw[->,thick] (a) edge[bend right] node[pos=.2,fill=white, inner sep=0pt] {$.5$} (s1);
    \draw[->,thick] (a) edge[bend left] node[pos=.2,fill=white, inner sep=0pt] {$.5$} (s3);
    
    \draw[->,thick] (c) -- node[pos=.3,fill=white, inner sep=0pt] {$.2$} (s2);
    \draw[->,thick] (c) -- node[fill=white, inner sep=0pt] {$.8$} (s3);

\end{tikzpicture}
		\caption{}
		\label{fig:mdp-for-pareto}
	\end{subfigure}
	\begin{subfigure}{0.49\linewidth}
		\centering
		 \begin{tikzpicture}[scale=.45]
    \draw[fill=black!10!white](0,0)--(0,3.2)--(2,2)--(3.2,0)--(0,0);

    \draw[->,thick] (-.5,0) -- (5,0);
    \draw[->,thick] (0,-.5) -- (0,5);

    \draw[thick] (4,-.25) -- (4,.25);
    \draw[thick] (-.25,4) -- (.25,4);

    \node (1a) at (4.2,-.6){$1$};
    \node (1b) at (-.6,4){$1$};

    \node (d1a) at (5.5,-.6){$\diamond \{s_1\}$};
    \node (d1b) at (-.6,5.5){$\diamond \{s_3\}$};

    \node (f1n) at (-1.35,3.2){$\textcolor{green!80!black}{(0,0.8)}$};
    \node (f2n) at (2.95,2.6){$\textcolor{green!80!black}{(0.5,0.5)}$};
    \node (f3n) at (2.8,-.6){$\textcolor{green!80!black}{(0.8,0)}$};

    \node[draw=green!80!black,circle,fill=green!80!black,inner sep=0pt,minimum size=5pt] (f1) at (0,3.2){};
    \node[draw=green!80!black,circle,fill=green!80!black,inner sep=0pt,minimum size=5pt] (f2) at (2,2){};
    \node[draw=green!80!black,circle,fill=green!80!black,inner sep=0pt,minimum size=5pt] (f3) at (3.2,0){};

    \draw[thick,green!80!black](0,3.2)--(2,2)--(3.2,0);

    % \draw[thick,blue] plot [smooth, tension=1] coordinates {(0,3.2)(2,2)(3.2,0)};

\end{tikzpicture}
		\caption{}
		\label{fig:pareto-front}
	\end{subfigure}
	
	\caption{(a) An example MDP. (b) Pareto frontier and achievable points for the multi-objective query $\objMO = (\textcolor{blue!90!white}{\{s_1\}},\textcolor{red!50!white}{\{s_3\}})$.}
	%\label{fig:enter-label}
\end{figure}
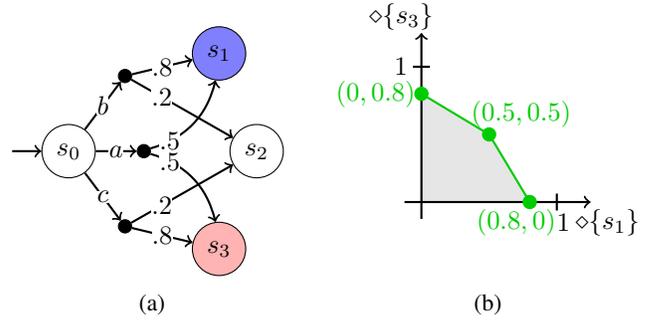
\begin{example}
	Consider the MDP in \cref{fig:mdp-for-pareto} with the multi-objective query 
	$\objMO$ = (\textcolor{blue!90!white}{$\{s_1\}$},\textcolor{red!50!white}{$\{s_3\}$}),
	i.e.\ we want to maximize our chances to reach the red and blue states.
    $s_2$ is not in any target set of our multi-objective query $\objMO$.
	Only action $b$ lets us achieve the point $(0.8,0)$, i.e.\ we reach a blue state with $0.8$ and a red state with probability $0$.
	Conversely, by playing only action $c$, we achieve the point $(0,0.8)$, and by playing action $a$, we achieve the point $(0.5,0.5)$.
	Alternatively, we can consider randomized strategies that schedule several actions with a certain probability. 
	For example, a strategy playing $a$ and $b$ with equal probability 
    %$a$ with $0.5$ and $b$ with $0.5$ 
    yields the point $(0.25, 0.65)$.
	All of these points are Pareto optimal, i.e.\ we cannot improve one dimension without worsening another. Thus, all of them are part of the solution of the multi-objective query.
	\qee
\end{example}
Given a multi-objective reachability query $\objMO$, we want to compute the set of achievable points, i.e.\ 
\[
\achievable_{\MDP}^{\objMO} \eqdef \{(p_1,\ldots,p_n) \mid \exists \strat \in \Strats. \forall i\in [1,n]. p_i \leq \probability_{\MDP}^{\strat}[\reach \targets_i]\},
\]
where $\reach \targets \eqdef \{ \infinitepath\in\Infinitepaths \mid \exists j\in \Naturals. \infinitepath_j \in \targets \}$ denotes the set of paths that reach a target set $\targets$.
Moreover, the \enquote{border} of this set is of interest, i.e.\ the Pareto optimal points; formally, the \emph{Pareto frontier} (also called Pareto curve~\cite{EKVY08}) is the set 
\[
\pf_{\MDP}^{\objMO} \eqdef \{ u \in \achievable_{\MDP}^{\objMO} \mid \neg\exists v\in\achievable_{\MDP}^{\objMO}. u\neq v \wedge u\leq v \}.
\]
\begin{example}
	Continuing the previous example, we show the Pareto frontier and achievable set in \Cref{fig:pareto-front}: The Pareto frontier is the green line and the set of achievable points is the area shaded grey, i.e.\ the polytope obtained from connecting the extreme points of the Pareto frontier and the origin point with 0 in every dimension.
	
	Even in this simple example, the Pareto frontier contains infinitely many points.
	We can obtain a finite representation by only considering the extreme points whose convex combinations make up the Pareto frontier.
	In our example, these extreme points are $\{(0.8,0),\,(0,0.8),\,(0.5,0.5)\}$.
	\qee
\end{example}

We denote by $\cp(\pf_{\MDP}^{\objMO})$ the set of all extreme points of a given Pareto frontier $\pf_{\MDP}^{\objMO}$.
Intuitively, these are the corners of the polytope that is the set of achievable points.
Formally, it is the set of all points on the Pareto frontier which are not a convex combination of other points on the Pareto frontier.

\begin{restatable}[Complexity of Pareto Frontier]{remark}{remarkPareto}
\label{rem:2-pareto-computation}
	While computing a Pareto frontier, in general, is $\PSPACE$-hard and requires exponential memory~\cite[Thm. 2]{DBLP:journals/fmsd/RandourRS17}, in the case of absorbing targets (i.e.\ for all $s\in\targets_i$ we have $\trans(s,a)(s)=1$ for all $a\in\act(s)$), it can be computed in $\PTIME$ by solving a multi-objective linear program and memoryless randomized strategies are sufficient~\cite[Thm. 3.2]{EKVY08}.
	We will reduce our problem to the latter case, i.e.\ with absorbing targets.
\end{restatable}
While the size of the exact Pareto frontier can be super-polynomial~\cite[Theorem 2.1]{EKVY08}, an $\varepsilon$-approximation containing polynomially many extreme points and can be computed in polynomial time~\cite[Cor. 3.5]{EKVY08}.
 
	\section{Mathematical characterization}\label{sec:3-title}
Throughout the paper, we often fix a weighted reachability (WR) objective $\objWR$ with the target set $\targets$ and reward function $\rew$,
and then talk about the associated vector $k$-vector of outcomes $\vec{o} = (o_1,\ldots,o_k)$.
We formalize this intuitive notion:

Firstly, this uses the fact that we overload $\obj$ to denote both the tuple $(\targets, \rew)$ describing the weighted reachability objective as well as the induced function that maps infinite paths to reals.
Secondly, the set of associated outcomes is formally defined as $\outcomes(\objWR) \eqdef \{ o \mid \exists s\in\targets. \rew(s) = o\} \cup \{0\}$.
Intuitively, the set of outcomes is the image set of the reward function $\rew$, i.e.\ all numbers that some state maps to.
Finally, we comment on a technicality: When applying a history-dependent strategy on an MDP, the set of states in the induced MC are all paths in the MDP~\cite[Def. 10.92]{BK08}. Thus, the set of target states becomes the set of all paths that contain a target state. 
We do not make this transformation explicit in the paper, but overload $\targets$ to always refer to the correct set of target states. Note that the set of outcomes is unaffected by this transformation.

\subsection{Alternative Definition of CPT-value}\label{sec:3-new-cpt-value-title}

The hitherto existing definition for CPT-value of an MDP (\cref{eq:2-cpt-val}) is quite complicated.
We propose a new, equivalent definition that is more intuitive.
Further, the new definition immediately leads to an algorithm for MCs (\cref{sec:4-mc-algo}).
Additionally, it allows us to establish a connection between the CPT-value in MDPs and multi-objective reachability queries (\Cref{sec:3-relation-cpt-multi-objective}).
We exploit this connection for the strategy complexity result and the algorithm for MDPs (\cref{sec:4-mdp-algo}).

\para{The Prospect Induced by a Markov Chain.}
The functions for computing both expected utility and CPT are defined on prospects, see \Cref{sec:2-main-cpt}.
Our goal is to interpret an MC (the result of fixing a strategy) as a prospect.
Then, we can evaluate strategies by applying the function given in \Cref{eq:2-cpt}, instead of mixing its internals into the value definition (as in \Cref{eq:2-cpt-val}).

So what is the prospect induced by a Markov chain $\MC$?
Intuitively, we want to compute the probability of reaching each of the absorbing target states; this yields a prospect because it gives us the probability for every outcome.
We need to consider two technicalities: 
Firstly, there can be multiple target states with the same reward, so their probabilities need to be summed up.
Secondly, a path in the MC can reach a state from which no target state is reachable anymore. Such a path eventually reaches a non-target \emph{bottom strongly connected component} (BSCC, see \cref{sec:2-graph-components}). 
The probability of these kinds of paths has to be summed up for the ourcome 0.
To formalize this, we define the set of states that obtain a certain outcome $o_i$, using a case distinction to take care of outcome~0.

\begin{equation}
    \obtainset(\MC,\objWR,o_i) \eqdef 
    \begin{cases}
    \{s\in\targets \mid \rew(s) = o_i\} &\mbox{ if } o_i\neq 0\\
    \bigcup_{C \in \BSCCs(\MC). C\cap \targets = \emptyset} C  &\mbox{ if } o_i=0
    \end{cases}
    \label{eq:3-obtainset}
\end{equation}

Using this and denoting the associated outcome vector of $\obj$ by $\vec{o}$, we define the induced prospect where the probability of each outcome is the probability to reach its $\obtainset$:
\begin{align}   
	&\prospect(\MC,\obj) \eqdef (\vec{o},\vec{p}), \nonumber \\
	&\hspace{0.5cm}\text{where for } 1\leq i \leq k: p_i = 
	\probability_{\MC}[\reach \obtainset(\MC, \obj,o_i)] \label{eq:3-induced-prospect}
\end{align} 

\begin{restatable}[Correctness of Induced Prospect]{lemma}{defsMakeSense}\label{lem:3-defs-make-sense}
	For every MC $\MC$ and WR-objective $\obj$, it holds that $\probability_{\MC}[\obj(\infinitepath)=o_i] = p_i$, where $(\vec{o},\vec{p}) = \prospect(\MC,\obj)$.
\end{restatable}
\begin{proof}
	This proof works by straightforward application of the definition of $\obtainset$ and induced prospect.
	By \Cref{eq:3-induced-prospect}, we have 
	$p_i = \probability_{\MC}[\reach \obtainset(\MC, \obj,o_i)]$.
	We now distinguish the cases whether the outcome $o_i$ is 0 or not.
	
	\textbf{Case 1:} For $o_i\neq0$, by \Cref{eq:3-obtainset}
	we have $\obtainset(\MC,\objWR,o_i) = \{s\in\targets \mid \rew(s) = o_i\}$.
	Since all non-target states have a reward of 0 and all target states are absorbing, 
    we can only obtain reward $o_i$ by reaching its $\obtainset$. 
	Thus, we conclude by stating:
	\[
	\probability_{\MC}[\obj(\infinitepath)=o_i] = 
	\probability_{\MC}[\reach \obtainset(\MC, \obj,o_i)].
	\]
	
	\textbf{Case 2:} For $o_i=0$, note that all paths not reaching a target state have reward 0.
	Recall we assume w.l.o.g. that no target state has a reward of 0, because removing such a state from the target set does not change the reward of a path reaching it.
	(Note that if we did not use this assumption, the $\obtainset$ for outcome 0 would also include all target states $s$ with $\rew(s)=0$; we make the assumption to avoid this additional notation in the definition of $\obtainset$.)
	From~\cite[Thm. 3.2 and Cor. 3.1]{de1997formal}, we have that every path in an MC will end up in a BSCC with probability 1.
	Using this and the fact that all target states are BSCCs (as they are absorbing by assumption), we have
	\begin{equation}
		1 = \probability_{\MC}[\reach \targets] + \probability_{\MC}[\reach \{C \in \BSCCs(\MDP^\strat) \mid C\cap \targets = \emptyset\}]\label{eq:3-bscc-with-prob-1}        
	\end{equation}
	Now we conclude as follows:
	\begin{align*}
		&\probability_{\MC}[\obj(\infinitepath)=0] = 
		1 - \probability_{\MC}[\reach \targets] \tag{By definition of $\obj$.}\\
		&=
		\probability_{\MC}[\reach \{C \in \BSCCs(\MDP^\strat) \mid C\cap T = \emptyset\}] \tag{By \Cref{eq:3-bscc-with-prob-1}} \\
		&= \probability_{\MC}[\reach \obtainset(\MC, \obj,0)]. \tag{By \Cref{eq:3-obtainset}}
	\end{align*}
	
\end{proof}

\para{New Equivalent CPT-value Definition.}\label{sec:3-new-cpt-value}
Using the prospect induced by an MC, we simplify the definition of CPT-value.
For an MDP $\MDP$, we fix a strategy to get an MC and then apply the function $\cptfun$ on the prospect induced by this MC.
The CPT-value is the supremum over all strategies.

\begin{align}
    \widehat{\CPTval}(\MDP,\obj) \eqdef\sup_{\strat\in\Strats}\cptfun(\prospect(\MDP^\strat,\obj)) \label{eq:3-cpt-value}
\end{align}

\begin{restatable}[Equivalence of Definitions]{theorem}{defsAlign}\label{thm:3-defs-align}
    The definitions of CPT-value in \Cref{eq:2-cpt-val,eq:3-cpt-value} coincide.
    Formally:
    \[
        {\CPTval}(\MDP,\obj) = \widehat{\CPTval}(\MDP,\obj).
    \]
\end{restatable}
\begin{proof}[Proof Sketch]
	This proof is mostly by definition.
	However, the definition of the CPT-function, in particular that of the decision weights, is quite involved, making it rather long, so we delegate it to \ifarxivelse{\cite[App. B-A]{techreport}}{\cref{app:3-proof-thm-defs-align}}.
	For instructive purposes, we also show that the equality holds in the special case of expected utility.
	The proof first unfolds the definition of $\widehat{\CPTval}$ and of the CPT-function $\cptfun$, and then applies \cref{lem:3-defs-make-sense} to show that the probabilities appearing in the induced prospect are exactly those used in \cref{eq:2-cpt-val}.
\end{proof}

Using \cref{thm:3-defs-align}, from now on we just write $\CPTval$ even when referring to \Cref{eq:3-cpt-value}.
Further, the theorem allows us to view the CPT-function as a black-box that is applied to a prospect and not worry about its internals; \Cref{sec:5-title} discusses the resulting advantages in more detail.

\para{Key contribution.}
The definition of the induced prospect in \cref{eq:3-induced-prospect} is the basis of all our results. 
Our algorithm for MCs (\cref{sec:4-mc-algo}) works by computing the induced prospect and then applying the CPT-function to it.
Our algorithm for MDPs intuitively works by concisely representing the set of all prospects that can be induced by some strategy in the MDP and then picking the optimum.

\subsection{Reduction to Stopping Markov Decision Processes}\label{sec:3-stopping}

For our reasoning in \Cref{sec:3-relation-cpt-multi-objective}, we want to employ the classical assumption that the MDPs we consider are \emph{stopping}, i.e.\ they reach a target state or a dedicated sink state almost surely.
Formally, we add a new state $\sink$ having a reward of \textbf{z}ero; 
then we require for every $\strat\in\Strats$ that 
$\probability_{\MDP}^{\strat}[\reach \targets \cup \{\sink\}] = 1$.

Intuitively, we merge all the BSCCs that cannot reach a target anymore (the $\obtainset$ for outcome 0) into a single state to more easily detect when a run obtains 0.
However, doing this in MDPs presents a complication:

\begin{example}\label{ex:3-stopping}
	Consider the simple MDP in \Cref{fig:3-stopping}.
    In state $s$, one can choose 
	%The state $s$ can either choose
    to self-loop or to go to $t$.
	Thus, depending on the strategy, either $s$ is in the $\obtainset$ of outcome 0, or it is not in any $\obtainset$.
	Note that in the first case, the resulting MC would not have the arrow for the action b, and $s$ would be a BSCC in the induced MC.
	\qee
\end{example}

\begin{figure}[t]
	\centering
	 \begin{tikzpicture}[scale=1.4]
    \node[draw,circle, minimum size=0.5cm] (s0) at (0,0){$s$};

    \node[draw,circle, minimum size=0.5cm] (t) at (1.5,0){$t$};
    
    \node (tc) at (2,0){$-5$};

    \draw[->,thick] (-.75,0) -- (s0);
    \draw[->,thick] (s0) --  node[above] {b} (t);
     \draw[->,thick] (s0) edge[loop above] node[above] {a} (s0);
     \draw[->,thick] (t) edge[loop above] node[above] {a} (t);

\end{tikzpicture}
	\caption{A simple example of a non-stopping MDP. Transition probabilities are omitted, as they are equal to 1.}
	\label{fig:3-stopping}
\end{figure}
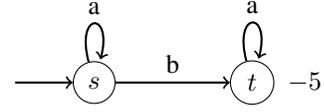

Such regions where under some strategies paths can cycle indefinitely are \emph{maximal end components} (MECs, see \cref{sec:2-end-components}).
For \emph{non-trivial} MECs (ones that are not targets, but can reach the target with positive probability), it depends on the strategy whether they should be merged into $\sink$ or not.
We solve this using the classical solution of \enquote{collapsing} MECs, see~\cite[Theorem 3.8]{de1997formal}, as also done in, e.g.,~\cite{AshokCDKM17,HM18}.

\begin{restatable}[Reduction to Stopping MDPs]{lemma}{stopping}\label{lem:3-stopping}
	For every MDP $\MDP$ and WR-objective~$\obj$, we can in polynomial time construct a stopping MDP $\QMDP$ with the same CPT-value, i.e.\ $\CPTval(\MDP,\obj) = \CPTval(\QMDP,\obj)$.
\end{restatable}
\begin{proof}[Proof Sketch]
	We follow the idea of~\cite{AshokCDKM17}, replacing every MEC with a single representative; its available actions are (i) all actions leaving the original MEC and (ii) a special action called $\stay$ leading to the sink state $\sink$.
	Intuitively, this allows to mimic every behaviour of the original MEC: If a strategy leaves the MEC in the original MDP, we can equivalently pick a distribution over the leaving actions, all of which are still available.
	If a strategy stays in the MEC forever in the original MDP (obtaining a reward of 0), this is represented by taking action $\stay$ and going to the dedicated sink state $\sink$.
	For the MDP in \Cref{fig:3-stopping}, this procedure removes the self-looping action and adds an action leading to the newly added state $\sink$.
	
	The resulting MDP $\MDP'$ is stopping because we eliminate every MEC in $\states \setminus (\targets \cup \{\sink\})$.
    Thus, it is impossible to surely cycle forever in any state but the absorbing targets or sink.
	The transformation takes polynomial time, as it computes MECs and iterates over them once.
	
	The full proof provided in \ifarxivelse{\cite[App. B-B]{techreport}}{\cref{app:3-stopping-details}} is more complex than similar proofs in related work~\cite{AshokCDKM17,HM18}, because the authors there relied on the assumption that memoryless deterministic strategies suffice for their objectives.
	In our case, the strategy complexity for $\CPTval$ is not yet solved, thus we also have to consider history-dependent randomized strategies.
\end{proof}

\subsection{Connection between CPT-value and Multi-objective Reachability Queries}
\label{sec:3-relation-cpt-multi-objective}

Intuitively, the new definition of CPT-value in \Cref{eq:3-cpt-value} leads to the following algorithm: 
For every strategy, evaluate the prospect of its induced MC, and take the supremum of these values.
As the set of all possible prospects is infinite, we now describe a way to finitely represent it.
Essentially, we construct a multi-objective reachability query from the given WR-objective $\obj$.
Consequently, the Pareto frontier for this query contains all possible probability distributions over target states and thus represents all prospects.

To transform the WR-objective to a multi-objective reachability query, we define the set of states that obtain a certain outcome.
We need to lift the definition of $\obtainset$ from MCs to MDPs.
All non-zero outcomes $o$ are obtained in target states that w.l.o.g.\ are absorbing (see \cref{sec:2-objective}).
For all strategies $\strat$, we have $\obtainset(\MDP,\obj,o)=\obtainset(\MDP^\strat,\obj,o)$.
However, the $\obtainset$ of 0 can vary depending on the strategy (see \Cref{ex:3-stopping}).
This is why in the previous \cref{sec:3-stopping}, we proved that w.l.o.g.\ we can restrict to stopping MDPs, where (i) states obtaining the reward 0 have been merged into a single sink state $\sink$; and (ii) $\targets \cup \{\sink\}$ is reached almost surely under all strategies.
We have $\obtainset(\MDP,\obj,0)=\{\sink\}$.
We define the multi-objective reachability query for a stopping MDP $\MDP$ and a WR-objective~$\obj$ (with outcomes numbered $o_1$ to $o_k$):
\begin{equation}
\label{eq:3-mo-query}
    \objMO(\MDP,\objWR) \eqdef (\obtainset(\MDP,\objWR,o_1),\ldots,\obtainset(\MDP,\objWR,o_k))
\end{equation}

We highlight that every point on the Pareto frontier of this transformed query $\pf_{\MDP}^{\objMO(\MDP,\objWR)}$ is a probability distribution.
This is because $\biguplus_{i=1}^k \obtainset(\objWR,o_k) = \targets \cup \{\sink\}$ and the MDP is stopping.
Thus, every point $(p_1,\ldots,p_k)$ on the Pareto frontier uniquely identifies a possible prospect. 
We formalize this:

\begin{restatable}{lemma}{core}
    \label{th:3-core-theorem}
    Fix a stopping MDP $\MDP$ and WR-objective $\obj$, with $\vec{o}$ being the vector of outcomes.
    The set of all prospects $\MDP$ can induce under some strategy $\strat$ corresponds to the Pareto frontier of the multi-objective reachability query $\objMO(\MDP,\obj)$.
    Formally:
    \[
    \Bigl\{ \prospect(\MDP^\strat,\objWR) \mid \strat\in\Strats \Bigr\} = 
    \Bigl\{ (\vec{o},\vec{p}) \mid \vec{p} \in \pf_{\MDP}^{\objMO(\MDP,\objWR)}\Bigr\}
    \]
\end{restatable}
\begin{proof}
	
	We show equality by proving inclusion in both directions.
	First, we fix an arbitrary strategy $\strat\in\Strats$ and show $\prospect(\MDP^\strat,\objWR) \in \Bigl\{ (\vec{o},\vec{p}) \mid \vec{p} \in \pf_{\MDP}^{\objMO(\MDP,\objWR)}\Bigr\}$. 
	
	Let $(\vec{o},\vec{q}) = \prospect(\MDP^\strat,\objWR)$ be the induced prospect.
	By definition of induced prospect, for every $i$ with $1 \leq i \leq k$, it holds that $q_i = \probability_{\MDP}^\strat[\reach \obtainset(\MDP, \obj,o_i)]$.
	Thus, by definition of $\objMO(\MDP,\objWR)$ and the achievable set, we know that $\vec{q} \in \achievable_{\MDP}^{\objMO(\MDP,\objWR)}$, as the target sets of the multi-objective query correspond exactly to the outcomes.
	
	Every point on the Pareto frontier $\pf_{\MDP}^{\objMO(\MDP,\objWR)}$ is a probability distribution.
	Hence, there is no point that is strictly better than $\vec{q}$.
	Thus, $\vec{q} \in \pf_{\MDP}^{\objMO(\MDP,\objWR)}$, which proves our goal.
	\medskip
	
	The other direction works analogously:
	Let $\vec{p}$ be a point on the Pareto frontier.
	Hence, a strategy $\strat$ exists that achieves exactly this distribution over the target sets,  which correspond to the outcomes.
	The prospect corresponding to this point is also in the set of all prospects induced by the MDP, or more formally: $(\vec{o},\vec{p}) \in \Bigl\{ \prospect(\MDP^\strat,\objWR) \mid \strat\in\Strats \Bigr\}$.
	%This concludes the proof.   
\end{proof}

We lift this connection between all possible prospects and the Pareto frontier to the CPT-value:

\begin{restatable}[Connection between CPT-value and Multi-objective Reachability]{theorem}{cptvalpf}\label{cor:3-cpt-val-pf}
    Fix an MDP $\MDP$ and WR-objective $\obj$ with $k$ outcomes, and let $\vec{o}$ be the vector of these outcomes.
    We can construct in polynomial time an MDP $\QMDP$ and a multi-objective reachability query $\objMO(\QMDP,\obj)$ of dimension $k$ with absorbing targets, such that the CPT-value of $\MDP$ equals the supremum over all prospects corresponding to points in the Pareto frontier $\pf_{\QMDP}^{\objMO(\QMDP,\objWR)}$.
    Formally:
    \[\CPTval(\MDP,\objWR) = \sup_{\vec{p} \in \pf_{\QMDP}^{\objMO(\QMDP,\objWR)}} \cptfun((\vec{o},\vec{p})).
    \]
\end{restatable}
\begin{proof}
Using \Cref{lem:3-stopping}, we construct $\QMDP$ from $\MDP$ in polynomial time, with $\CPTval(\MDP,\obj) = \CPTval(\QMDP,\obj)$.
Then: 
\begin{align*}
	\CPTval(\QMDP,\objWR) &\eqdef \sup_{\strat\in\Strats}\cptfun(\prospect(\QMDP^\strat,\obj)) \tag{By \Cref{eq:3-cpt-value}}\\
	&=\sup_{ \vec{x} \in \{\prospect(\QMDP^\strat,\objWR) \mid \strat\in\Strats \} } \cptfun(\vec{x}) \\
	&=\sup_{\vec{p} \in \pf_{\QMDP}^{\objMO(\QMDP,\objWR)}} \cptfun((\vec{o},\vec{p})) \tag{By \Cref{th:3-core-theorem}}
\end{align*}
\end{proof}

\para{Key contribution.} The definition of the multi-objective query in \cref{eq:3-mo-query}, together with the technical stopping assumption from \Cref{sec:3-stopping}, allows us to reduce the problem of computing the CPT-value to finding the CPT-optimal point on the Pareto frontier.
Solving this optimization problem is the only remaining step to obtain the algorithm for MDPs in \cref{sec:4-mdp-algo}.

Moreover, this reduction immediately yields an upper bound on strategy complexity:
The optimal strategy achieves a point on the Pareto frontier $\pf_{\MDP}^{\objMO(\MDP,\objWR)}$ by \Cref{cor:3-cpt-val-pf}, and every point on the Pareto frontier can be achieved using a memoryless randomized strategy by \cref{rem:2-pareto-computation}. Together with an example that randomized strategies can achieve a strictly higher CPT-value than deterministic ones (see \ifarxivelse{\cite[App. B-C]{techreport}}{\Cref{app:examples-rand}}), we obtain the following strategy complexity:
\begin{restatable}[Strategy Complexity]{corollary}{stratcompl}
	\label{thm:3-strat-compl}
	Memoryless randomized strategies are necessary and sufficient to achieve the optimal CPT-value in an MDP $\MDP$ with a WR-objective~$\obj$. 
\end{restatable}
	\section{Algorithm and computational complexity}\label{sec:4-title}
We present algorithms to compute the CPT-value of an MC and an MDP, as well as their complexity. 
We conclude the section by demonstrating how to extract an $\varepsilon$-optimal strategy.

\subsection{Algorithms and Complexity for MC}\label{sec:4-mc-algo}

\begin{algorithm}[t]
	\caption{CPT-value in Markov Chains}\label{alg:4-cpt-mc}
	\begin{algorithmic}[1]
		\Require MC $\MC$ and WR-objective $\obj$, CPT-function $\cptfun$
		\Ensure CPT-value $\CPTval(\MC,\obj)$
		\State $\outcomevector\gets \text{unique}^\uparrow(\outcomes(\objWR) )$\label{line:4-outcomes}\Comment{Outcomes}
		\State $\nu^\MC\gets \mathsf{stat\_distr}(\MC)$ \label{line:4-lim-dis} \Comment{Stationary distribution}
		\State $\probabilities \gets \left( \sum_{s \in \obtainset(\MC,o_i,\obj)} \nu^{\MC}(s) \right)_{i\in[1,k]}$\label{line:4-probs}
		\State \textbf{Return} $\cptfun((\outcomevector,\probabilities))$ \label{line:4-cpt-function}\Comment{CPT-function}
	\end{algorithmic}
\end{algorithm}

We present our algorithm for computing the CPT-value of an MC in \cref{alg:4-cpt-mc}:
Lines \ref{line:4-outcomes}-\ref{line:4-probs} compute the prospect of the MC, $\prospect(\MC,\obj) = (\outcomevector,\probabilities)$, before applying the CPT-function on it in Line~\ref{line:4-cpt-function} (see \cref{sec:2-prelim-cpt-eval} for its pseudocode).
We compute the prospect by first creating the increasingly sorted vector of all outcomes in \Cref{line:4-outcomes}, numbered $o_1, \ldots, o_k$.
Next, we find the \emph{stationary distribution} $\nu : \states \to [0,1]$ of the MC, e.g.\ using the efficient algorithms of~\cite{DBLP:conf/tacas/Meggendorfer23}.
This is a probability distribution over states; as all target states are absorbing, the stationary distribution on these coincides with the probability of reaching them, see~\cite[Eq. (3)]{DBLP:conf/tacas/Meggendorfer23}.
We use them to aggregate the probabilities of reaching an outcome and construct the probability distribution over outcomes in \Cref{line:4-probs}.

\begin{example}
	We exemplify \cref{alg:4-cpt-mc} on the MC obtained by always playing $a_1$ in the MDP from \Cref{fig:motivating-example-MDP}; we essentially remove action $a_2$ as well as states $s_3$ and $s_4$ from the figure.
    The algorithm first obtains the vector of outcomes in the MC, namely 0 and 20.
    The stationary distribution of this MC is $(s_0\rightarrow 0, s_1\rightarrow 0.95, s_2\rightarrow 0.05)$. %, as there is no way to stay in $s_0$, and $s_1$ and $s_2$ are absorbing.
	This defines the prospect $\prosp_{a_1}=[0:0.05, 20: 0.95]$. 
	Now, we can evaluate this prospect using the CPT-function.
	We remark that while this running example is very simple (as a sink state is reached after one step), the algorithm is applicable to arbitrary MCs. \qee
\end{example}

\begin{lemma}[Correctness]
	\label{lem:4-algo-cpt-mc}
	Given a Markov chain $\MC$ and a WR-objective $\obj$, 
	%and an oracle for $\cptfun$ 
	\cref{alg:4-cpt-mc} computes $\CPTval(\MC)$ up to the precision of the given oracle for the CPT-function $\cptfun$.
\end{lemma}
\begin{proof}
	On absorbing states, the stationary distribution $\nu^\MC$ coincides with the reachability probability~\cite[Eq. (3)]{DBLP:conf/tacas/Meggendorfer23}.
	Thus, we have for all $i$ with $1\leq i\leq k$ that $\sum_{s \in \obtainset(\MC,\obj,o_i)} \nu^{\MC}(s) =     \probability_{\MC}[\reach \obtainset(\MC, \obj,o_i)] = \probability_{\MC}[\obj(\infinitepath)=o_i]$.
	Consequently, the prospect computed by the algorithm is the induced prospect of the given MC, formally $(\outcomevector,\probabilities) = \prospect(\MC,\objWR)$.
	Then correctness follows immediately from correctness of the new definition of the CPT-value (\Cref{thm:3-defs-align}).

Since the CPT-function typically involves utility and weighting functions that contain roots (see \ifarxivelse{\cite[App. A-A]{techreport}}{\cref{app:2-u-and-w}}), our algorithm can only return a value as precise as the $\cptfun$-oracle.
\end{proof}

\begin{lemma}[Complexity]
	\label{lem:4-algo-cpt-mc-complexity}
	Given an MC $\MC$ and WR-objective $\obj$, \Cref{alg:4-cpt-mc} terminates in polynomial time.
\end{lemma}
\begin{proof}
	The stationary distribution can be computed in polynomial time by solving a system of equations~\cite[Thm. 2.5]{kulkarnimodeling-second-edition}.
	Additionally, the algorithm consists of at most two iterations over all states in the MC: 
	once to find all unique rewards and once to accumulate the probabilities for each reward.
	Finally, as the CPT-function (which uses elementary functions in $\util$ and $\weight$) is part of our input, we assume that we have an oracle for it and can evaluate it in constant time.
\end{proof}

\begin{theorem}[MC complexity]
\label{th:4-mc}
    The CPT-value approximation problem for MCs is in $\PTIME$.
\end{theorem}

\subsection{Algorithm and Complexity for MDP}\label{sec:4-mdp-algo}

\begin{algorithm}[t]
	\caption{CPT-value in Markov Decision Processes}\label{alg:4-cpt-mdp}
	\begin{algorithmic}[1]
		\Require MDP $\MDP$, WR-objective $\objWR$, CPT-function $\cptfun$, precision $\varepsilon$
		\Ensure Approximation of the CPT-value $\CPTval(\MDP,\obj)$
		\State $\QMDP\gets \mathsf{make\_stopping}(\MDP,\objWR)$ \Comment{Transform to stopping}\label{line:stopping}
		\State $\objMO\gets \objMO(\QMDP,\objWR)$ \Comment{Transform to multi-objective}\label{line:mo}
		\State $\pf \gets $ $\varepsilon$-approximation of $\pf_{\QMDP}^{\objMO}$
		\label{line:pareto}
		\State \textbf{Return} $\varepsilon$-approximation of $\sup_{\probabilities \in \pf} \cptfun((\outcomevector,\probabilities))$\label{line:optimitzation} 
	\end{algorithmic}
\end{algorithm}

\para{Algorithm Overview.}
\cref{alg:4-cpt-mdp} shows how to compute the CPT-value of an MDP.
It transforms the MDP to a stopping one with the same CPT-value (Line~\ref{line:stopping}, see \Cref{lem:3-stopping}), and then translates the given WR-objective to a multi-objective query (Line~\ref{line:mo}, see \cref{eq:3-mo-query}). 
Line~\ref{line:pareto} uses a standard procedure to compute (an approximation of) the Pareto frontier (e.g.\ from \cite{EKVY08}).
By \cref{th:3-core-theorem}, this Pareto frontier represents all possible prospects.
Thus, in Line~\ref{line:optimitzation}, we can rewrite the problem as an optimization problem which finds the $\cptfun$-optimal point on the approximation of the Pareto frontier. 
Solving this optimization problem yields the CPT-value.

\begin{example}
	The MDP $\MDP$ from \cref{ex:2-running} is already stopping. 
	We have $\objMO = (F_{-5},F_0,F_{20},F_{50}) =$ $(\{s_3\}, \{s_2\}, \{s_1\}, \{s_4\})$, where each set corresponds to one possible outcome.
    Thus, $\cp(\pf_{\MDP}^{\objMO})=\{(0.95,0.05,0,0),(0,0.05,0.44,0.51)\}$.
    Here, the Pareto frontier forms a line between the two extreme points, so we can easily find the optimal CPT-value. 
    In general, the last step is more involved.
	\qee
\end{example}

\para{Key Difficulties.}
There are two obstacles to overcome to prove the correctness of the algorithm:
Firstly, since the CPT-function is non-convex (see \ifarxivelse{\cite[App. C-A]{techreport}}{\cref{app:4-pseudoconvex}}), Line~\ref{line:optimitzation} cannot be computed simply.
Non-convex optimization (also called global optimization) is a very difficult problem.
We refer to~\cite{matthiesen2019efficient} for a discussion of available algorithms, improving upon~\cite{tuy1998convex}.
Most importantly, known algorithms for non-convex optimization (i)~are always at least worst-case exponential~\cite[Chap. 3]{matthiesen2019efficient}, (ii)~often require the objective function to be monotonic~\cite[Chap. 4]{matthiesen2019efficient}, which CPT is not (see \ifarxivelse{\cite[App. A]{techreport}}{\cref{app:4-cpt-increasing}}), and (iii)~often only converge in the limit, without guaranteeing a precision in finite time~\cite[Chap. 3.2.1]{matthiesen2019efficient}.
For solving the optimization problem in Line~\ref{line:optimitzation}, we thus take a brute-force approach: We prove Lipschitz-continuity of the CPT-function and discretize the set of all possible probability vectors with sufficient precision, allowing us to prove an exponential bound on the runtime.

Secondly, bounding the imprecision of \Cref{alg:4-cpt-mdp} is intricate because there are multiple sources of approximation error: (i)~from solving the optimization problem in Line~\ref{line:optimitzation}, (ii)~from approximating the Pareto frontier in Line~\ref{line:pareto}, and (iii)~from propagating the latter error through the CPT-function.
We show how to bound this error in \cref{lem:4-algo-mdp}.
We remark that, apart from approximation errors, the correctness of our approach is immediate from \cref{cor:3-cpt-val-pf}.

\para{Section Outline.} 
In the remainder of this section, we analyse \cref{alg:4-cpt-mdp}.
For correctness, \cref{lem:4-cpt-lip-cont} proves the main requisite: Lipschitz-continuity of the CPT-function, allowing us to solve the optimization problem (\cref{lem:4-solve-opt-new}) and prove the correctness (\cref{lem:4-algo-mdp}), and to overcome both of the key difficulties mentioned above.
\cref{lem:4-algo-mdp-complexity} proves a bound on the runtime of \Cref{alg:4-cpt-mdp}, required to establish the complexity of the CPT-value approximation problem for MDPs in \cref{th:4-mdp} and \cref{cor:4-fpt}.
Finally, we also remark on how to compute a CPT-optimal strategy.

\begin{restatable}[Lipschitz-continuity of CPT-function]{lemma}{cptLc}
	\label{lem:4-cpt-lip-cont}
	The CPT function $\cptfun$ is Lipschitz-continuous with respect to the probability vector $\probabilities$, i.e.\ for every $\varepsilon \in \Reals$ and vector of outcomes $\outcomevector\in \Reals^k$, we have that if two vectors $\probabilities, \probvectorQ \in [0,1]^k$ satisfy $\abs{\probabilities - \probvectorQ} \leq \varepsilon$ (using point-wise difference and $L_2$-norm), then 
	\[
		\abs{\cptfun((\outcomevector,\probabilities)) - \cptfun((\outcomevector,\probvectorQ))} \leq \varepsilon \cdot \lipCPT,
	\]
	where $\lipCPT \eqdef \util^* \cdot \max(\lipWp,\lipWm) \cdot (2k^2+k)$ is the Lipschitz-constant of the CPT-function, which depends on the outcome with the largest absolute utility $\util^* \eqdef \max_{i \in [1,k]} \abs{\util(o_i)}$, the Lipschitz-constants $\lipWp,\lipWm$ of the probability weighting functions $\pweight$ and $\mweight$, and the number of outcomes $k$.
\end{restatable}
\begin{proof}[Proof Sketch]
This very technical proof (\ifarxivelse{\cite[App. C-C]{techreport}}{\cref{app:4-cpt-lipschitz}}) proceeds by straightforward unfolding of the CPT-function $\cptfun$.
\end{proof}

\begin{lemma}[Solving the Optimization Problem]\label{lem:4-solve-opt-new}
	Given a Pareto frontier $\pf$ and an outcome vector $\outcomevector$, we can $\varepsilon$-approximate 
	$\sup_{\probabilities \in \pf} \cptfun((\outcomevector,\probabilities))$, and the computation takes time 
	$(\frac {\lipCPT \cdot \sqrt{k}} \varepsilon)^k \cdot \mathcal{O}(n\log(n))$, where $n = \max(\cp(\pf),2^k)$.
\end{lemma}
\begin{proof}
	The Pareto frontier is a set of $k$-dimensional vectors, i.e.\ $\pf \subseteq \Reals^k$. Further, all components of the vectors are bounded to be in $[0,1]$. 
	We can discretize this space of vectors into hypercubes with side length $\frac \varepsilon {\lipCPT \cdot \sqrt{k}}$.
	The resulting number of hypercubes is finite, namely $(1 / \frac \varepsilon {\lipCPT \cdot \sqrt{k}})^k = (\frac {\lipCPT \cdot \sqrt{k}} \varepsilon)^k$.
	
	The largest distance between two vectors in a hypercube is less than the length of its diagonal, i.e.\ $\sqrt{k}$ times its side length, resulting in $\frac \varepsilon {\lipCPT}$.
	Consequently, the difference between the CPT-function evaluated on two different vectors in the same hypercube is at most $\frac \varepsilon {\lipCPT} \cdot \lipCPT = \varepsilon$.
	Using this, we proceed as follows: 
	For each hypercube, compute its intersection with $\pf$, and if it is non-empty, select one vector from the intersection.
	The CPT-value of this vector is $\varepsilon$-close to the CPT-value of all other vectors in the hypercube. 
	Consequently, taking the maximum over all hypercubes yields a point that is $\varepsilon$-close to the desired quantity $\sup_{\probabilities \in \pf} \cptfun((\outcomevector,\probabilities))$.
	
	The runtime of this approach is dominated by the number of hypercubes we consider, $(\frac {\lipCPT \cdot \sqrt{k}} \varepsilon)^k$. 
	For each hypercube, computing the intersection takes time in $\mathcal{O}(n\log(n))$, where $n$ is the number of vertices of the considered polytope~\cite[Lem.~3]{DBLP:conf/iccS/TereshchenkoCF13}. 
	Thus, $n$ is the maximum of $\cp(\pf)$ and $2^k$, depending on whether the Pareto frontier or the the hypercube has more vertices.
	Evaluating the CPT-function takes constant time because we assumed that we have an oracle for it.
\end{proof}

While the proposed algorithm is exponential, its asymptotic complexity is in line with the best known algorithms for non-convex optimization, which are at least exponential in the number of variables~\cite[Chap. 3]{matthiesen2019efficient}, i.e.\ in our case also exponential in $k$.
Still, in practice, we recommend using the Branch-and-Bound algorithm~\cite[Chap. 3.1]{matthiesen2019efficient}, which, intuitively, incrementally discretizes the set of possible probability vectors, potentially avoiding to investigate all hypercubes.
Using it requires specifying a \enquote{bound operation}, which is possible using the Lipschitz-continuity of the CPT-function.

\begin{restatable}[Correctness and Error Estimation]{lemma}{algomdp}
	\label{lem:4-algo-mdp}
	For an MDP $\MDP$, a WR-objective $\objWR$ with $k$ different outcomes, a CPT-function $\cptfun$, and a given precision $\varepsilon$, \cref{alg:4-cpt-mdp} computes the CPT-value $\CPTval(\MDP,\obj)$ up to a precision of $\varepsilon \cdot (1+\lipCPT)$, where $\lipCPT$ is the Lipschitz-constant of the CPT-function.
\end{restatable}
\begin{proof}
	Let $v$ be the output of \cref{alg:4-cpt-mdp}.
	The following chain of equations proves our goal:	
	\begin{align*}
		&\abs{\CPTval(\MDP,\obj) - v} \\
		&\leq 
		\abs{\CPTval(\MDP,\obj) - \sup_{\probsappr \in \pf} \cptfun((\outcomevector,\probsappr)) + \varepsilon_{\text{opt}}} \tag{I}\\
		&= \abs{\sup_{\probabilities\in\pf_{\QMDP}^{\objMO}} \cptfun(\outcomevector,\probabilities) - \sup_{\probsappr\in\pf} \cptfun(\outcomevector,\probsappr)} + \varepsilon_{\text{opt}} \tag{II}\\
		&\leq \varepsilon_{\pf} \cdot \lipCPT + \varepsilon_{\text{opt}} \tag{III}\\
		&= \varepsilon \cdot (1+\lipCPT) \tag{IV}
	\end{align*}
	
	We explain every step separately:
	
	\noindent\textbf{(I):} By  \cref{lem:4-solve-opt-new}, the algorithm returns an $\varepsilon_{\text{opt}}$-approximation of $\sup_{\probsappr \in \pf} \cptfun((\outcomevector,\probsappr))$, where the subscript opt indicates the error from non-convex optimization.
	
	\noindent\textbf{(II):} 
	By \cref{cor:3-cpt-val-pf}, we have $\CPTval(\MDP,\obj) = \sup_{\vec{p} \in \pf_{\QMDP}^{\objMO(\QMDP,\objWR)}} \cptfun((\vec{o},\vec{p}))$.
	Note that we replaced $\MDP$ with its stopping variant $\QMDP$ constructed in Line~\ref{line:stopping}.
	
	\noindent\textbf{(III):}
    $\varepsilon_{\pf}$ is the error introduced by approximating the frontier, i.e.\ the maximum distance between any two vectors $\probabilities$ and $\probsappr$.
	This error is propagated through the CPT-function $\cptfun$, multiplying it with the Lipschitz-constant $\lipCPT$, as defined in \cref{lem:4-cpt-lip-cont}.
	
	\noindent\textbf{(IV):}
	We choose $\varepsilon_{\text{opt}}$ and $\varepsilon_{\pf}$ to be equal to the precision $\varepsilon$ given in the claim.
    In practice, using different precision for these two sources of error could be advantageous.	
\end{proof}

\begin{lemma}[Runtime of \cref{alg:4-cpt-mdp}]
	\label{lem:4-algo-mdp-complexity}
	Given an MDP $\MDP$, a WR-objective $\objWR$ with $k$ different outcomes, 
	a CPT-function $\cptfun$, and a precision $\varepsilon$, 
	\cref{alg:4-cpt-mdp} has the following runtime:
	The first three steps are polynomial in the size of the MDP $\MDP$, i.e.\ $\abs{\MDP}$, and the inverted precision $\frac 1 \varepsilon$.
	The last step, solving the optimization problem, takes the larger of the following two times:
	$\mathcal{O}\left((\frac{k^2 \cdot \util^* \cdot \sqrt{k} } \varepsilon)^k \cdot (2^k \cdot k)\right)$ or $\mathcal{O}\left((\frac {k^2 \cdot \util^* \cdot \sqrt{k} } \varepsilon)^k \cdot \text{poly}(|\MDP|,\frac{1}{\varepsilon})\right)$, where $\util^*$ is the largest absolute utility appearing in $\objWR$.
\end{lemma}
\begin{proof}
	We show the runtime of each algorithm step.
	\begin{itemize}
		\item Line~\ref{line:stopping} takes polynomial time in the size of the MDP $\abs{\MDP}$ by \cref{lem:3-stopping}, and the resulting stopping MDP $\QMDP$ has at most $\abs{\states} + 1$ many states.
		This upper bound is proven as follows: The MEC quotient construction (see \ifarxivelse{Def. B-2}{\cref{def:3-mec-quotient}} in \ifarxivelse{\cite[App. B-B]{techreport}}{\cref{app:3-stopping-details}}) removes every MEC and adds a single new state for it. Since a MEC is a non-empty set of states, this transformation cannot increase the number of states, but at most leave it unchanged (if every MEC in the original MDP has size 1 or there are no MECs).
		Additionally, the MEC quotient adds the new sink state $\sink$, which is the reason for the one additional state. 
		\item Line~\ref{line:mo} takes linear time in the number of states $\abs{\states}$, as it iterates once over all target states.
		\item Line~\ref{line:pareto} approximates the Pareto frontier in time polynomial in the size of the MDP and the precision, $\text{poly}(|\MDP|,\frac{1}{\varepsilon})$ \cite[Cor. 3.5]{EKVY08}. The result has polynomial many points in the input, $|\cp(\pf)|=\text{poly}(|\MDP|,\frac{1}{\varepsilon})$.
		\item Line~\ref{line:optimitzation} takes time
		$(\frac {\lipCPT \cdot \sqrt{k}} {\varepsilon_{\text{opt}}})^k \cdot \mathcal{O}(n\log(n))$, where $n = \max(\cp(\pf),2^k)$ (by \cref{lem:4-solve-opt-new}).
		It remains to discuss two things: 
		Firstly, $\lipCPT$ depends on the given WR-objective, namely quadratically on the number of outcomes $k$ and linearly on the outcome with the largest possible utility $\util^* \eqdef \max_{i \in [1,k]} \abs{\util(o_i)}$ (whereas the Lipschitz-constants of the probability weighting functions are not relevant for asymptotic complexity).
		Thus, we replace $\lipCPT$ with $k^2 \cdot \util^*$.
		Secondly, we make a case distinction on whether $n$ is $2^k$, or whether it is $\cp(\pf)$.
		In the former case, the runtime is in
		$\mathcal{O}\left((\frac {k^2 \cdot \util^* \cdot \sqrt{k} } \varepsilon)^k \cdot (2^k \cdot k)\right)$.
		In the latter case, the runtime is in $\mathcal{O}\left((\frac {k^2 \cdot \util^* \cdot \sqrt{k} } \varepsilon)^k \cdot \text{poly}(|\MDP|,\frac{1}{\varepsilon})\right)$. \qed
	\end{itemize}
\renewcommand{\qedsymbol}{} 
\end{proof}

\begin{theorem}[MDP complexity]
\label{th:4-mdp}
    The CPT-value approximation problem for MDPs can be solved in time polynomial in the size of the MDP $\abs{\MDP}$ and exponential in the number of outcomes $k$, where the basis of the exponentiation is a polynomial containing $k$, the largest absolute utility $\util^*$, and the desired precision $\varepsilon$.
\end{theorem}
\begin{proof}
	We run \cref{alg:4-cpt-mdp} with precision $\frac \varepsilon \lipCPT$ in order to achieve overall precision $\varepsilon$ (\cref{lem:4-solve-opt-new}).
	Inserting this in \cref{lem:4-algo-mdp-complexity} yields a runtime in $\mathcal{O}\left((\frac{k^4 \cdot (\util^*)^2 \cdot \sqrt{k} } \varepsilon)^k \cdot (2^k \cdot k)\right)$ or $\mathcal{O}\left((\frac{k^4 \cdot (\util^*)^2 \cdot \sqrt{k} } \varepsilon)^k \cdot \text{poly}(|\MDP|,\frac{1}{\varepsilon})\right)$.
	Note that Line~\ref{line:optimitzation} dominates the runtime of the algorithm.
\end{proof}

\begin{remark}[Lower runtime bound for the algorithm]\label{rem:lower-compl-bound}
    It is unlikely that with an approach like ours --- using non-convex optimization to find the optimal point on the Pareto frontier --- the worst-case runtime can be polynomial in $k$.
    The CPT-function is necessarily non-convex, as both the utility and the probability weighting functions have a non-convex \enquote{S-shape}.
    One of the largest known classes of optimization functions, where such an optimization problem still is polynomial, is the class of \emph{pseudo-convex} functions \cite{pseudo-convex}.
    However, as shown by the counter-example in \ifarxivelse{\cite[App. C-A]{techreport}}{\Cref{app:4-pseudoconvex}}, the CPT-function is not pseudo-convex. 
    Thus, the existence of a polynomial-time algorithm for the computation is non-trivial and challenging.
\end{remark}

The complexity of the CPT-value approximation problem largely depends on factors not related to the MDP, namely number of outcomes, largest utility, and precision.
However, the number of outcomes $k$ is typically small compared to the size of the MDP, and similarly, we can often bound the largest appearing utilities in the objective (using knowledge of the objective) and the precision (commonly $10^{-6}$).
Thus, the problem is fixed-parameter tractable~\cite[Def. 1]{grohe1999descriptive}.

\begin{corollary}
\label{cor:4-fpt}
    The CPT-value approximation problem for MDPs is (a) in EXPTIME and (b) fixed-parameter tractable when fixing the precision $\varepsilon$, the number of outcomes $k$, and the largest absolute utility $\util^*$.
\end{corollary}
\begin{proof}
	Note that $n^n = 2^{n \cdot \log_2(n)}$. Thus, we can rewrite the runtime bound from \cref{th:4-mdp} to be of the form $\mathcal{O}(2^{\text{poly}(n)})$, with $n$ being the size of the input, proving (a). For (b), observe that the runtime already is in the form of a product of a function of fixed parameters and a polynomial of the size of the MDP, thus satisfying~\cite[Def. 1]{grohe1999descriptive}.
\end{proof}

\begin{remark}[Extracting a Strategy]\label{sec:4-strats}
We are interested not only in the optimal CPT-value, but also a CPT-optimal strategy.
Since the optimal CPT-value corresponds to a point on the Pareto frontier, we can use \cite[Corollary 3.5]{EKVY08} to obtain the corresponding strategy.
Additionally, we have to reverse the transformation to a stopping MDP (see \cref{sec:3-stopping}) to obtain a strategy on the original MDP.
This is formally described at the end of \ifarxivelse{\cite[App. B-B]{techreport}}{\cref{app:3-stopping-details}}.
\end{remark}
	\section{Robustness of Our Result}\label{sec:5-title}

\subsection{Variations of the CPT-function $\cptfun$}\label{sec:5-cpt-variations}
We have always used $\cptfun$ as a black-box. Our results only require that it is Lipschitz-continuous (\cref{lem:4-cpt-lip-cont}). 
In particular, our theory still applies for all reasonable choices of utility function $\util$ and probability weighting functions $\pweight$, $\mweight$ (which are used for the decision weight $\decweight$). 
This is relevant because these functions differ between individuals, as they are affected by, e.g., cultural differences \cite{rieger2017estimating}.
Further, if an improved variant of CPT is developed (much like CPT~\cite{CPT92} replacing the original prospect theory~\cite{prospectTheory79}), our methods are most probably still applicable: 
This new function will still map prospects to reals and probably be Lipschitz-continuous, so that all our reasoning applies.

\subsection{Variations of the Objective}

\paragraph{Minimizing Instead of Maximizing.}
So far, our formulations only look at the supremum, as we deem it intuitive to maximize rewards.
Naturally, you could also use infimum to minimize the obtained reward, for example, to find the minimal possible CPT-value in an MDP (aka. worst-case scenario).
All our techniques still apply by replacing the supremum with an infimum.
In particular, computing prospects and the Pareto frontier are independent of the use of maximum/minimum.

\paragraph{Variants of Weighted Reachability.}
We can change the outcome that is associated with not reaching any target state.
Throughout the paper, we assumed that this corresponds to no change.
Depending on the modeled situation, it could be reasonable to say that not reaching any target should be penalized by giving some reward that is less than every other outcome.
Our theory can easily be extended to this case: Modify the definition of prospect induced by an MC to assign this penalty outcome to all non-target BSCCs.
With this change, all results follow without any further modification.
In particular, we can compute the Pareto frontier representing all prospects and solve the optimization problem on it.

\paragraph{Mean Payoff.}\label{sec:mp}

Given a reward function $r\colon\states\rightarrow\Rationals$, a mean payoff objective (a.k.a.\ long-run average reward) computes the average reward obtained in the limit: 
\[
	\text{MP}(\rho)=\lim_{n\rightarrow\infty} \left(\frac{1}{n}\sum_{i=0}^{n-1}r(\rho_i)\right)
\]
In \ifarxivelse{\cite[App. D]{techreport}}{\Cref{app:5-mp}}, we provide a reduction from mean payoff to weighted reachability under CPT, based on a similar reduction under expected utility given in~\cite{KM18-cvar}.
Thus, all our results carry over to mean payoff.

\begin{restatable}[Reduction of Mean Payoff to Weighted Reachability]{theorem}{mpreduction}\label{cor:cpt-mp-is-cpt-wr}
	We can transform an MDP $\MDP$ with a mean payoff objective $\objMP$ to another MDP $\MDP^f$ with a weighted reachability objective $\obj^{WR}$ in polynomial time, such that
	$$\CPTval(\MDP,\objMP)=\CPTval(\MDP^f,\obj^{WR})$$
\end{restatable}
\begin{proof}[Proof Sketch]
	For mean payoff, the decisive parts of an MDP are the infinite tails, which occur within a MEC of the MDP.
	For each MEC, it is possible to compute the mean payoff within it.
	Using this, one can add an additional state for each MEC in the MDP, that represents the idea of staying in the MEC forever (similar to the $\stay$-action in \Cref{sec:3-stopping}) and obtaining its mean payoff.
	A strategy can choose whether it wants to leave the MEC, making the MEC transient and irrelevant to the mean payoff; or it can choose to stay within the MEC, obtaining its mean payoff.
	
	Each of the newly added states has a reward equal to the mean payoff of the respective MEC. 
    We can view it as an MDP with a weighted reachability objective that computes the same prospects as the MDP with the mean payoff objective.
\end{proof}

	\section{Conclusion}\label{sec:6-title}

We have thoroughly investigated MDPs with CPT, in particular (i) giving a new intuitive definition of CPT-value in an MDP, (ii) providing an algorithm that correctly approximates the CPT-value for every precision $\varepsilon>0$, and (iii) establishing results on the strategy and computational complexity.
These results are based on connecting the CPT-value with multi-objective reachability queries.

We list several interesting directions for future work:
Implementing the algorithm allows to evaluate both its performance as well as the quality of the decisions that are recommended.
To improve the former, we can employ a branch-and-bound algorithm, utilizing the information that the Pareto frontier only intersects with a fraction of the possible hypercubes. Further, outcomes are ordered, so some prospects are clearly dominated by others.
This reduces the number of points that have to be considered in the optimization problem.
Regarding the quality of decisions recommended, it is also interesting to compare CPT with other theories of risk like mean-variance~\cite{markowitz1991foundations}, conditional value-at-risk~\cite{rockafellar2000optimization}, or entropic risk~\cite{DBLP:journals/fs/FollmerS02}.

Our definition of CPT-value employs a kind of \enquote{static} semantics, in the sense that we fix a complete strategy for the whole MDP and then view it as a prospect.
A \enquote{dynamic} semantics is also of interest, where the current state and the history that led up to it affect the decision; such a dynamic prospect theory has recently been proposed~\cite{tymula2023dynamic}.
This comes with several complications and design choices to be made, as discussed in~\cite[App. C]{wakker2010prospect}: In particular, under non-expected utility, it makes a difference whether we propagate rewards forward or backward through the graph.
It is open whether the nested definition of CPT-value in~\cite{lin2013dynamic,lin2018probabilistically} coincides with one of these views or is something else entirely.

Finally, extending our theory to multi-player environments (cf.~\cite{DBLP:journals/tac/EtesamiSMP18,danis2023multi}) would allow to lift the unrealistic assumption that all players behave rationally (i.e.\ according to expected utility theory).
Similarly, it is unrealistic to assume that all transition probabilities are fully known. 
We refer to~\cite{kahneman1982judgment} for a psychological analysis of how people judge probabilities and to~\cite[Part III]{wakker2010prospect} for a thorough discussion of this setting and how CPT extends to it.

	\bibliographystyle{IEEEtran}
	\bibliography{ref}

    % ---- Appendix ----
    \ifarxivelse{}{
	\appendices
	\crefalias{section}{appendix} 
	\crefalias{subsection}{appendix} 
	\crefalias{subsubsection}{appendix}
	\section{Additional Details on Preliminaries}\label{app:2-extended-prelims}

\subsection{Common Utility and Probability Weighting Functions}\label{app:2-u-and-w}

This appendix discusses the utility function $\util$ and probability weighting functions $\weightP$ and $\weightM$.
All three functions are increasing (f is \emph{increasing} if for $x\geq x'$, $f(x)\geq f(x')$) and continuous.

In \Cref{fig:2-functions-app}, we depict commonly used utility and probability weighting functions.
	The utility function is concave for gains and convex for losses~\cite[Sec. 1.3]{CPT92}.
	This matches the intuition given above that utility increases less than linearly with high gains and that losses are weighed more than gains.
	The probability weighting functions overestimate smaller probabilities and underestimate higher probabilities.
	For more details, we refer to \cite[Chapter 7.2]{wakker2010prospect}.
	
	\cite[Section 2.3]{CPT92} present standard functions for the utility function~$\util$ and the probability weighting functions $\pweight$, $\mweight$:
	$$\util(x)=\begin{cases}
		x^\alpha & \text{if}\;x\geq0\\
		-\lambda(-x)^\beta& \text{if}\;x<0
	\end{cases}$$
	and
	$$\weightP(x)=\frac{x^\gamma}{(x^\gamma+(1-x)^\gamma)^\frac{1}{\gamma}}$$
	$$\weightM(x)=\frac{x^\delta}{(x^\delta+(1-x)^\delta)^\frac{1}{\delta}}$$
	where it was found in~\cite{CPT92} that good parameters are $\alpha=\beta=0.88$, $\lambda=2.25$, $\gamma=0.61$, $\delta=0.69$.
	
	Note that we changed the parameter $\alpha$ of the utility-function $\util$ in \cref{fig:2-u-pic} to $0.3$ (similar to \cite[Figure 1]{barberis2013thirty}) to make the S-shape clearly visible.
	The actual shape of $\util$ with the recommended value of $0.88$ is shown in \cref{fig:true-u}.
	
	We highlight that all functions are non-convex. Further, they all use roots, and thus applying them can result in irrational numbers.

\begin{figure}
	\centering
	\begin{subfigure}{0.45\linewidth}
		\centering
		\includegraphics[width=\linewidth]{u-0.5-s.png}
		\caption{}
		\label{fig:2-u-pic}
	\end{subfigure}
	\begin{subfigure}{0.45\linewidth}
		\centering
		\includegraphics[width=\linewidth]{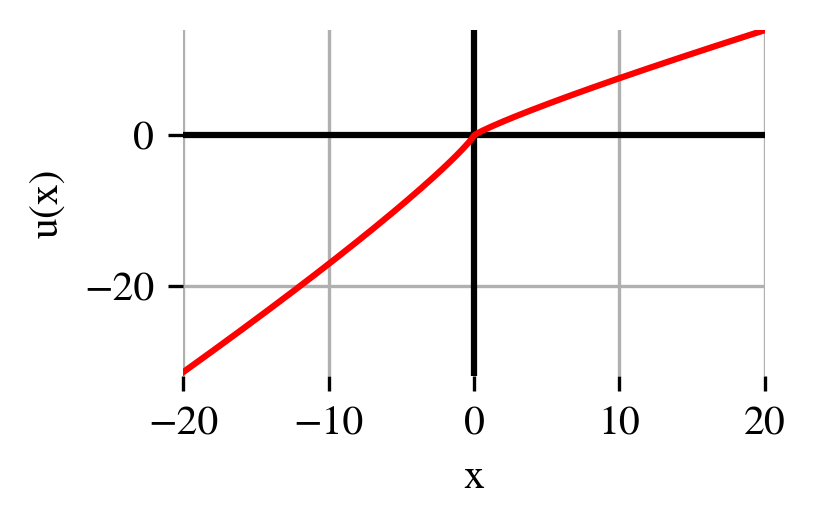}
		\caption{}
		\label{fig:true-u}
	\end{subfigure}
    
	\begin{subfigure}{0.45\linewidth}
		\centering   
		\includegraphics[width=\linewidth]{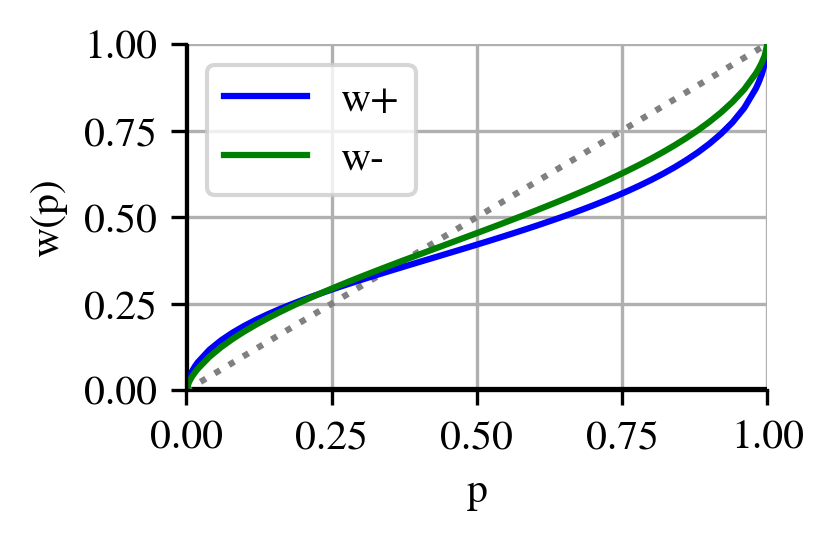}   
		\caption{}
		\label{fig:2-w-pic}
	\end{subfigure}
	\caption{(a,b) Standard function for the utility function $\util$, with $\alpha=0.3$ or $\alpha=0.8$, where (a) was also shown in \cref{fig:2-functions}. (c) Standard functions for the probability weighting functions $\pweight$ and $\mweight$. 
	}
	\label{fig:2-functions-app}
\end{figure}

\subsection{Definition of CPT-function}\label{app:2-CPT}

We list the major components of CPT. Recall that we fix a prospect $\vec{x} = (\outcomevector,\vec{p})$ with $k$ different outcomes.
\begin{itemize}
	\item CPT uses a utility function $\util: \outcomeset\to \Reals$ with the same purpose as in expected utility theory.
	\item CPT splits the outcomes $\outcomevector$ into gains and losses because the relation to the current status quo greatly affects human preference.
	In the following, let $o_j = 0$ be the outcome indicating that nothing changes. 
	As the outcomes are ordered, all outcomes with $1\leq i < j$ are losses ($o_i < 0$) and all outcomes with $j<i\leq k$ are gains ($o_i > 0$).
	\item CPT uses probability weighting functions $\weightP$ and $\weightM$, both of type $[0,1] \to [0,1]$. These modify the actual probabilities since humans tend to under-estimate large probabilities and over-estimate small probabilities, see \Cref{app:2-u-and-w}. 
	It uses two functions because depending on whether humans are weighting the probability of a gain ($\weightP$) or a loss ($\weightM$), their perception of probabilities changes.
	\item CPT multiplies each outcome with a \emph{decision weight}.
	While in expected utility theory, this weight is just the individual probability of each outcome,
	in CPT, we instead consider the \emph{marginal contribution} of the outcome. 
	In other words, the decision weight of a positive outcome $o_i$ depends on the difference between gaining at least $o_i$ and strictly more than $o_i$; dually, for a negative outcome, the decision weight depends on the difference between gaining at most $o_i$ and strictly less than $o_i$.
	Formally, we define the gain rank $\gainrank$~\cite[Def 5.4.1]{wakker2010prospect}, the probability of getting strictly more than $o_i$; and the loss rank $\lossrank$~\cite[Eq. 7.6.3]{wakker2010prospect}, the analog with less instead of more.
	\begin{align*}
		\gainrank(\probabilities,i) \eqdef \sum_{m = i+1}^k p_m \quad \quad
		\lossrank(\probabilities,i) \eqdef \sum_{m = 1}^{i-1} p_m
	\end{align*}
	Note that this definition assumes that the outcome vector is ordered increasingly.
	For further information on these concepts, we refer to the chapters in \cite{wakker2010prospect} surrounding their definition and to the analysis in the original CPT paper~\cite[Sec. 1.1]{CPT92}.
	
	Using gain rank and loss rank, we define the decision weight $\decweight(\probabilities, i)$ of an outcome $o_i$ as follows:
	
\end{itemize}

	\begin{align*}
		\decweight_i(\vec{x}) \eqdef
		\begin{cases}
			\weightP(p_i + \gainrank(\probvector,i)) - \weightP(\gainrank(\probvector,i))     & (o_i > 0)\\
			0 & (o_i = 0)\\
			\weightM(p_i + \lossrank(\probvector,i)) - \weightM(\lossrank(\probvector,i))& (o_i < 0)
		\end{cases}
	\end{align*}

\begin{algorithm}[t]
	\caption{CPT on a single prospect}
\label{alg:get-cpt-recurse}
\begin{algorithmic}[1]
	\Require $\prosp=[o_1:p_1,\ldots,o_k:p_k]$ \Comment{Prospect}
	\State Sort $\prosp$ descending in the outcomes
	\State $c,cp,cm\gets 0$
	\For{$(o,p)$ in $\prosp$}
	\If{$o>0$}
	\State $c$ += $\util(o)(\pweight(p+cp) - \pweight(cp))$
	\State $cp$ += $p$
	\Else
	\State $c$ += $\util(o)(\mweight(p+cm) - \mweight(cm))$
	\State $cm$ += $p$
	\EndIf
	\EndFor
	\Ensure $c$ 
\end{algorithmic}
\end{algorithm}

Using these components of CPT, we now give the formal definition of the CPT-function, see~\cite[Eq. 9.2.2]{wakker2010prospect}.
Note that the label is the same as in the main body.
This is done for all equations and statements that appear both in the main body and the appendix to avoid having different labels for the same content.

\begin{equation}\label{eq:app-2-cpt}
\cptfun(\vec{x}) = \sum_{i=1}^{k}  \util(o_i) \cdot \decweight(\prosp, i) \tag{\ref{eq:2-cpt}}
\end{equation}

For our complexity results, we assume that we are given an oracle for the (potentially elementary) functions $\mweight$,$\pweight$, and $\util$, i.e.\ they can be evaluated in constant time.

For the convenience of an interested reader, we also show an algorithm for evaluating the CPT-function on a given prospect in \cref{alg:get-cpt-recurse}.
In this algorithm, we ease the computation by accumulating the probabilities in a helper variable.

\subsection{Extended Example \cref{ex:2-running}}\label{app:2-exs}
	Consider again \cref{ex:intro-motivate} where we have a coupon for a bet and can choose between two bets.
	Formally, we have an MDP $\MDP=(\states,\initstate,\act,\trans)$ with $\states=\{s_0,s_1,s_2,s_3,s_4\}$, and an initial state $s_0$, as depicted in \cref{fig:app-motivating-example-MDP}. 
	We omit the self looping actions on states $s_1, \ldots, s_4$.
	For $s_0$, we have $\act(s_0)=\{a_1,a_2\}$, % and $\act(s_i)=\emptyset$ for $i\in[2,4]$.
	with each action representing one of the bets: 
    
    The safe bet resulting in a 95\% chance of 20€, $\trans(s_0,a_1)=(s_1\rightarrow 0.95, s_2\rightarrow 0.05)$, 
    and the risky bet resulting in a 51\% chance of getting 50€, and a 41\% chance of loosing 5€, $\trans(s_0,a_2)=(s_2\rightarrow 0.05, s_3\rightarrow 0.44, s_4\rightarrow 0.51)$.
	We have a set of target states $\targets=\{s_1,s_3,s_4\}$ and the reward function  $r(s_1)=20$, $r(s_3)=-5$, $r(s_4)=50$.

	\begin{figure}
		\centering
		 \begin{tikzpicture}[xscale=1,yscale=0.7]
    \node[draw,circle, minimum size=0.5cm] (s0) at (0,0){$s_0$};

    \node[draw,circle,fill=black,inner sep=0pt,minimum size=5pt] (a1) at (-1.5,0){};
   
    \node[draw,circle,fill=black,inner sep=0pt,minimum size=5pt] (a2) at (1.7,0){};

    \node[draw,circle, minimum size=0.5cm] (s1) at (-3,-1.2){$s_1$};
    \node[draw,circle, minimum size=0.5cm] (s2) at (-0.5,-1.2){$s_2$};
    \node[draw,circle, minimum size=0.5cm] (s3) at (1,-1.2){$s_3$};
    \node[draw,circle, minimum size=0.5cm] (s4) at (3,-1.2){$s_4$};

    \node (s1c) at (-3,-2){$20$};
    \node (s2c) at (-0.5,-2){$0$};
    \node (s3c) at (1,-2){$-5$};
    \node (s4c) at (3,-2){$50$};

    \draw[->,thick] (-.75,0.75) -- (s0);
    \draw[->,thick] (s0) -- node[fill=white, inner sep=1pt] {$a_1$} (a1);
    \draw[->,thick] (s0) -- node[fill=white, inner sep=1pt] {$a_2$} (a2);

    \draw[->,thick] (a1) -- node[fill=white, inner sep=1pt] {$.95$} (s1);
    \draw[->,thick] (a1) -- node[fill=white, inner sep=1pt] {$.05$} (s2);

    \draw[->,thick] (a2) -- node[fill=white, inner sep=1pt] {$.05$} (s2);
    \draw[->,thick] (a2) -- node[fill=white, inner sep=1pt] {$.44$} (s3);
    \draw[->,thick] (a2) -- node[fill=white, inner sep=1pt] {$.51$} (s4);

\end{tikzpicture}
		\caption{Running example as MDP (repeat of \Cref{fig:motivating-example-MDP} for convenience).}
		\label{fig:app-motivating-example-MDP}
	\end{figure}
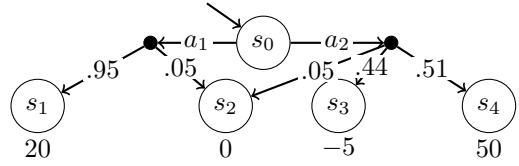
	
	\begin{figure}
		\centering
		 \begin{tikzpicture}
    \node[draw,circle, minimum size=0.9cm] (s0) at (0,0){$s_0$};

    \node[draw,circle,fill=black,inner sep=0pt,minimum size=5pt] (a1) at (.75,3.5){};
   
    \node[draw,circle,fill=black,inner sep=0pt,minimum size=5pt] (a2) at (.75,-3.5){};

    \node[draw,circle, minimum size=0.9cm] (s1) at (2.5,5){$s_1$};
    \node[draw,circle, minimum size=0.9cm] (s2) at (2.5,2){$s_2$};
    \node[draw,circle, minimum size=0.9cm] (s3) at (2.5,-2){$s_3$};
    \node[draw,circle, minimum size=0.9cm] (s4) at (2.5,-5){$s_4$};

    %\node (s1c) at (2.75,1.5){$20$};
    %\node (s2c) at (2.75,.5){$0$};
    %\node (s3c) at (2.75,-.5){$-5$};
    %\node (s4c) at (2.75,-1.5){$50$};

    \draw[->,thick] (-.75,0) -- (s0);
    \draw[->,thick] (s0) -- node[fill=white, inner sep=1pt] {$a_1$} (a1);
    \draw[->,thick] (s0) -- node[fill=white, inner sep=1pt] {$a_2$} (a2);

    \draw[->,thick] (a1) -- node[fill=white, inner sep=1pt] {$.95$} (s1);
    \draw[->,thick] (a1) -- node[fill=white, inner sep=1pt] {$.05$} (s2);

    \draw[->,thick] (a2) -- node[left,fill=white, inner sep=1pt] {$.05$} (s2);
    \draw[->,thick] (a2) -- node[fill=white, inner sep=1pt] {$.44$} (s3);
    \draw[->,thick] (a2) -- node[fill=white, inner sep=1pt] {$.51$} (s4);

    \node[draw,circle, minimum size=0.9cm] (s5) at (6,7.5){$s_5$};
    \node[draw,circle, minimum size=0.9cm] (s6) at (6,6.5){$s_6$};
    \node[draw,circle, minimum size=0.9cm] (s7) at (6,5.5){$s_7$};
    \node[draw,circle, minimum size=0.9cm] (s8) at (6,4.5){$s_8$};
    \node[draw,circle, minimum size=0.9cm] (s9) at (6,3.5){$s_9$};
    \node[draw,circle, minimum size=0.5cm] (s10) at (6,2.5){$s_{10}$};
    \node[draw,circle, minimum size=0.5cm] (s11) at (6,1.5){$s_{11}$};
    \node[draw,circle, minimum size=0.5cm] (s12) at (6,.5){$s_{12}$};
    \node[draw,circle, minimum size=0.5cm] (s13) at (6,-.5){$s_{13}$};
    \node[draw,circle, minimum size=0.5cm] (s14) at (6,-1.5){$s_{14}$};
    \node[draw,circle, minimum size=0.5cm] (s15) at (6,-2.5){$s_{15}$};
    \node[draw,circle, minimum size=0.5cm] (s16) at (6,-3.5){$s_{16}$};
    \node[draw,circle, minimum size=0.5cm] (s17) at (6,-4.5){$s_{17}$};
    \node[draw,circle, minimum size=0.5cm] (s18) at (6,-5.5){$s_{18}$};
    \node[draw,circle, minimum size=0.5cm] (s19) at (6,-6.5){$s_{19}$};
    \node[draw,circle, minimum size=0.5cm] (s20) at (6,-7.5){$s_{20}$};

    \node[draw,circle,fill=black,inner sep=0pt,minimum size=5pt] (a11) at (4,6){};
    \draw[->,thick] (s1) -- node[fill=white, inner sep=1pt] {$a_1$} (a11);
    \draw[->,thick] (a11) -- node[fill=white, inner sep=1pt] {$.95$} (s5);
    \draw[->,thick] (a11) -- node[fill=white, inner sep=1pt] {$.05$} (s6);

    \node[draw,circle,fill=black,inner sep=0pt,minimum size=5pt] (a12) at (4,5){};
    \draw[->,thick] (s1) -- node[fill=white, inner sep=1pt] {$a_2$} (a12);
    \draw[->,thick] (a12) -- node[fill=white, inner sep=1pt] {$.05$} (s6);
    \draw[->,thick] (a12) -- node[fill=white, inner sep=1pt] {$.44$} (s7);
    \draw[->,thick] (a12) -- node[fill=white, inner sep=1pt] {$.51$} (s8);

    \node[draw,circle,fill=black,inner sep=0pt,minimum size=5pt] (a21) at (4,2.5){};
    \draw[->,thick] (s2) -- node[fill=white, inner sep=1pt] {$a_1$} (a21);
    \draw[->,thick] (a21) -- node[fill=white, inner sep=1pt] {$.95$} (s9);
    \draw[->,thick] (a21) -- node[fill=white, inner sep=1pt] {$.05$} (s10);

    \node[draw,circle,fill=black,inner sep=0pt,minimum size=5pt] (a22) at (4,1.5){};
    \draw[->,thick] (s2) -- node[fill=white, inner sep=1pt] {$a_2$} (a22);
    \draw[->,thick] (a22) -- node[fill=white, inner sep=1pt] {$.05$} (s10);
    \draw[->,thick] (a22) -- node[fill=white, inner sep=1pt] {$.44$} (s11);
    \draw[->,thick] (a22) -- node[fill=white, inner sep=1pt] {$.51$} (s12);

    \node[draw,circle,fill=black,inner sep=0pt,minimum size=5pt] (a31) at (4,-1.5){};
    \draw[->,thick] (s3) -- node[fill=white, inner sep=1pt] {$a_1$} (a31);
    \draw[->,thick] (a31) -- node[fill=white, inner sep=1pt] {$.95$} (s13);
    \draw[->,thick] (a31) -- node[fill=white, inner sep=1pt] {$.05$} (s14);
    
    \node[draw,circle,fill=black,inner sep=0pt,minimum size=5pt] (a32) at (4,-2.5){};
    \draw[->,thick] (s3) -- node[fill=white, inner sep=1pt] {$a_2$} (a32);
    \draw[->,thick] (a32) -- node[fill=white, inner sep=1pt] {$.05$} (s14);
    \draw[->,thick] (a32) -- node[fill=white, inner sep=1pt] {$.44$} (s15);
    \draw[->,thick] (a32) -- node[fill=white, inner sep=1pt] {$.51$} (s16);

    \node[draw,circle,fill=black,inner sep=0pt,minimum size=5pt] (a41) at (4,-5){};
    \draw[->,thick] (s4) -- node[fill=white, inner sep=1pt] {$a_1$} (a41);
    \draw[->,thick] (a41) -- node[fill=white, inner sep=1pt] {$.95$} (s17);
    \draw[->,thick] (a41) -- node[fill=white, inner sep=1pt] {$.05$} (s18);

    \node[draw,circle,fill=black,inner sep=0pt,minimum size=5pt] (a42) at (4,-6){};
    \draw[->,thick] (s4) -- node[fill=white, inner sep=1pt] {$a_2$} (a42);
    \draw[->,thick] (a42) -- node[fill=white, inner sep=1pt] {$.05$} (s18);
    \draw[->,thick] (a42) -- node[fill=white, inner sep=1pt] {$.44$} (s19);
    \draw[->,thick] (a42) -- node[fill=white, inner sep=1pt] {$.51$} (s20);

    \node (s5c) at (6.75,7.5){$40$};
    \node (s6c) at (6.75,6.5){$20$};
    \node (s7c) at (6.75,5.5){$15$};
    \node (s8c) at (6.75,4.5){$70$};
    \node (s9c) at (6.75,3.5){$20$};
    \node (s10c) at (6.75,2.5){$0$};
    \node (s11c) at (6.75,1.5){$-5$};
    \node (s12c) at (6.75,0.5){$50$};
    \node (s13c) at (6.75,-0.5){$15$};
    \node (s14c) at (6.75,-1.5){$-5$};
    \node (s15c) at (6.75,-2.5){$-10$};
    \node (s16c) at (6.75,-3.5){$45$};
    \node (s17c) at (6.75,-4.5){$70$};
    \node (s18c) at (6.75,-5.5){$50$};
    \node (s19c) at (6.75,-6.5){$45$};
    \node (s20c) at (6.75,-7.5){$100$};

\end{tikzpicture}
		\caption{\cref{ex:intro-motivate} as MDP. 
			We omit the self looping actions on states $s_5, \ldots, s_{20}$.
			Note that many of the states $s_5, \ldots, s_{20}$ share the same reward and could be merged.
			We show all of them to avoid confusion in the picture.
		}
		\label{fig:app-motivating-example-MDP-sequential}
	\end{figure}
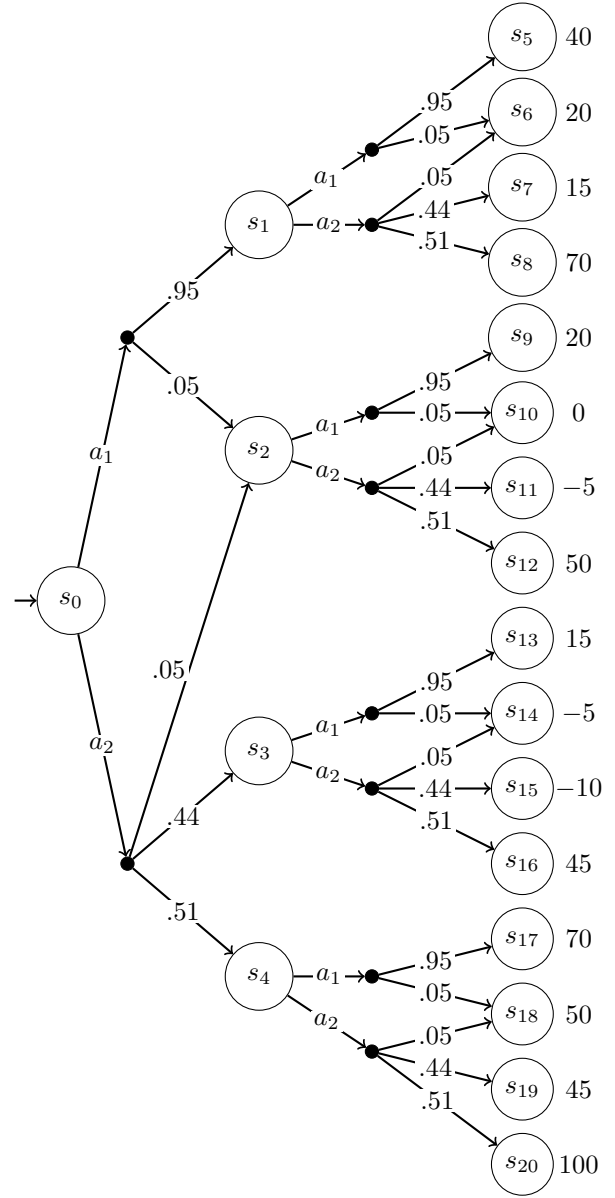
	
	Recall that we assume w.l.o.g. that no target state has a reward of 0; thus, $s_2$ is not a target state.
	One could make it a target state with reward 0, but as every path that does not reach a target gets reward 0, this is unnecessary.
	In \cref{sec:3-title}, our notation will be simplified by assuming that no target state has a reward of 0. 
	
	The actions correspond to the following prospects: 
    $\prosp_1=[0:0.05,\;20:0.95]$ and 
    $\prosp_2=[-5:0.44,\;0:0.05,\;50:0.51]$.
	
	Using the standard choices for $\util$, $\weightM$, and $\weightP$, we compute 
    %$\cptfun(\prosp_1) \approx 5.40$ and $\cptfun(\prosp_2)\approx 3.40$.
    $\cptfun(\prosp_1) \approx 11.07$ and $\cptfun(\prosp_2) \approx 10.19$.
	On the other hand, 
    we get $\eu(\prosp_1)=19$ and $\eu(\prosp_2)=23.3$, 
    with $\util(x)=x$.
	Thus, expected utility theory recommends choosing $a_2$ while CPT predicts that most humans prefer $a_1$. 	

    We can extend the example further by adding a second choice for a bet.
    After betting once, we bet again with the same conditions.
    This results in a tree-like MDP shown in \cref{fig:app-motivating-example-MDP-sequential}.
    The possible strategies are now a combination of playing $a_1$ or $a_2$ twice: 
    $(a_1,a_1)$ or $(a_1,a_2)$ or $(a_2,a_1)$ or $(a_2,a_2)$.

    For each strategy, we get a prospect:
    \begin{itemize}
    \item$\prosp_{11} = [0:0.025,\;20:0.095,\;40:0.9025]$
    \item$\prosp_{12} = [-5~:~0.022,\; 0~:~0.0025,\;15~:~0.418,\;$
    
    $20~:~0.0475,\;50~:~0.0255,\;70:0.4845]$
    \item$\prosp_{21}$ = $\prosp_{12}$.
    \item$\prosp_{22} = [-10:0.1936,\;-5:0.044,\;0:0.0025,\;$
    
    $45:0.4488,\;50:0.051,\;100:0.2601]$
    \end{itemize}
    Evaluating each of them with expected utility and CPT yields:
    
    \begin{tabular}{c|c|c}
         Prospect & $\eu$ & $\cptfun$ \\
         \hline
         $\prosp_{11}$& 38& 21.79 \\
         $\prosp_{12}$ &42.3&\textbf{21.99}\\
         $\prosp_{21}$&42.3&\textbf{21.99}\\
         $\prosp_{22}$&\textbf{46.6}&21.18\\
    \end{tabular}
    
    Note that the expected utility can be computed exactly, and the CPT-value in this case is an approximation.
    This shows that expected utility chooses the risky bet twice (resulting in an expected utility of $46.6$), whereas CPT recommends the safe and risky bet each once (resulting in a CPT value of $21.99$).
    
    We highlight that it is only possible to compute this manually and simply because of the tree structure.
    We highlight that this example can only be computed manually because of the simple tree structure. 
    More complex structures, in particular loops, require computing the Pareto frontier as shown in \cref{alg:4-cpt-mdp}.
	\qee

\subsection{Definition of CPT-value in previous works}\label{app:2-cptval}

Intuitively, we replace the expectation operator in \cref{eq:2-eu-val} with the CPT-function.
We use the functions $\weightP$ and $\weightM$, and compute the gain rank and loss rank, see \Cref{app:2-CPT}.
Further, we define the set of possible outcomes as follows: $\outcomes(\objWR) \eqdef \{ o \mid \exists s\in\targets. \rew(s) = o\} \cup \{0\}$.
This is useful, because CPT sums over all possible outcomes.
Furthermore, we split positive and negative outcomes, so we write
$\outcomes^+(\objWR) \eqdef \{ o\in \outcomes(\objWR)\mid o>0\}$, and analogously $\outcomes^-(\objWR)$.  We define the CPT-value as follows:
\vspace{-1em}
{\par\nopagebreak\small%
	\begin{align}
		\scriptsize
		\tag{\ref{eq:2-cpt-val}}
		&\CPTval(\MDP,\obj) \eqdef \sup_{\strat\in\Strats} \bigg(\\
		&
		\sum_{o \in \outcomes^+(\obj)} \util(o) \cdot \left( \weightP\left( \probability_{\MDP}^{\strat}(\obj(\infinitepath) \geq o) \right) - \weightP\left( \probability_{\MDP}^{\strat}(\obj(\infinitepath) > o) \right)\right) + \nonumber \\
		&\sum_{o \in \outcomes^-(\obj)} \util(o) \cdot \left(\weightM\left( \probability_{\MDP}^{\strat}(\obj(\infinitepath) \leq o) \right) - \weightM\left( \probability_{\MDP}^{\strat}(\obj(\infinitepath) < o) \right)\right) \nonumber
		\bigg)
	\end{align}
}

Note that this definition mimics the style of other papers considering MDPs with CPT, e.g.~\cite[Chap. 3.7]{DBLP:journals/ftml/A022}, but is different because, unlike us, they consider infinitely many outcomes.
We refer to~\cite[App. 9.7]{wakker2010prospect} for a derivation of their formulation.

We provide a new and equivalent (shown in \Cref{thm:3-defs-align}) definition of CPT-value in \Cref{sec:3-new-cpt-value}, based on viewing an MC as a prospect.
Note that this is exactly what we have done in \Cref{ex:2-running}: View the MC resulting from choosing an action as a prospect and evaluate that using expected utility or the CPT-function.

	\section{Additional details for Section~\ref{sec:3-title}}\label{app:3-title}
\subsection{Equivalence of CPT-Value Definitions --- Proof of \Cref{thm:3-defs-align}}\label{app:3-proof-thm-defs-align}

Before we prove \cref{thm:3-defs-align}, we first explicitly prove that the definitions align in the special case of expected utility.
This is instructive and further supports the correctness of the definition of induced prospect.
We define the expected value as follows, using the $\eu$-function on prospects as defined in \cref{sec:2-prelim-eu}:
\begin{equation}
	\widehat{\val}(\MDP,\obj) \eqdef \sup_{\strat\in\Strats}\eu(\prospect(\MDP^\strat,\obj)) \label{eq:3-eu-value}
\end{equation}

\begin{lemma}
	The definitions of expected value in \cref{eq:2-eu-val,eq:3-eu-value} coincide. Formally:
	\[
	{\val}(\MDP,\obj) = \widehat{\val}(\MDP,\obj)
	\]
\end{lemma}
\begin{proof}
	The following chain of equations proves our goal.
	\begin{align*}
		\val(\MDP,\obj) &\eqdef \sup_{\strat\in\Strats} \Expectation_{\MDP}^{\strat}[\util(\obj(\infinitepath))] \tag{By \Cref{eq:2-eu-val}}\\
		&=\sup_{\strat\in\Strats} \sum_{o \in \outcomes(\objWR)} \util(o) \cdot \probability_{\MDP}^{\strat}[\obj(\infinitepath)=o] \tag{By definition of expectation}\\
		&=\sup_{\strat\in\Strats} \sum_{i=1}^k \util(o_i) \cdot \probability_{\MDP}^{\strat}[\obj(\infinitepath)=o_i] \tag{By definition of $\vec{o}$}\\
		&=\sup_{\strat\in\Strats} \sum_{i=1}^k \util(o_i) \cdot p_i, \text{where } (\vec{o},\vec{p}) = \prospect(\MDP^\strat,\obj) \tag{By \Cref{lem:3-defs-make-sense}}\\
		&=\sup_{\strat\in\Strats} \eu(\prospect(\MDP^\strat,\obj)) \tag{By \Cref{eq:2-eu-app}}
	\end{align*}    
\end{proof}

\defsAlign*

\begin{proof}
    The proof for CPT is analogous to the one for expected utility, but cumbersome due to the complicated definition of the CPT-function.
    Throughout the proof we let $\vec{x} = (\vec{o},\vec{p}) = \prospect(\MDP^\strat,\obj)$.
    Further, we split the outcomes $\vec{o}$ into positive and negative outcomes as described in \Cref{app:2-CPT}, i.e.\ we let $\vec{o}_j = 0$ be the outcome that nothing changes. 
    As the outcomes are ordered increasingly, all outcomes with $1\leq i < j$ are losses ($\vec{o}_i < 0$) and all outcomes with $j<i\leq k$ are gains ($\vec{o}_i > 0$).
    
    \textbf{Observation 1:} By definition of the decision weights, we have $\pi_j(\vec{x}) = 0$. Thus, we also have $\util(o_j) \cdot \pi_j(\vec{x}) = 0$.
    
    \textbf{Observation 2:} We have $\outcomes^-(\obj) = \{o_i \mid 1\leq i < j\}$ and $\outcomes^+(\obj) = \{o_i \mid j<i\leq k\}$.

    \medskip
    
	We start from the new definition of \Cref{eq:3-cpt-value}.
    \begin{align*}
        \widehat{\CPTval}(\MDP,\obj) &\eqdef\sup_{\strat\in\Strats}\cptfun(\prospect(\MDP^\strat,\obj)) \tag{By \Cref{eq:3-cpt-value}}\\
        &=\sup_{\strat\in\Strats}\sum_{i=1}^k \util(o_i) \cdot \pi_i(\vec{x})\tag{By \Cref{eq:2-cpt}}\\
        &=\sup_{\strat\in\Strats}\sum_{i=1}^{j-1} \util(o_i) \cdot \pi_i(\vec{x}) +\sum_{i=j+1}^{k} \util(o_i) \cdot \pi_i(\vec{x})\tag{By splitting the sum and using \textbf{Observation 1}}\\
        &=\sup_{\strat\in\Strats}\sum_{o\in\outcomes^-(\obj)} \util(o_i) \cdot \pi_i(\vec{x}) + \\&\phantom{=\sup_{\strat\in\Strats}}\sum_{o\in\outcomes^+(\obj)} \util(o_i) \cdot \pi_i(\vec{x})\tag{By \textbf{Observation 2}}
    \end{align*}

    The final line of this sequence of equations is almost the same as the definition of CPT-value in \Cref{eq:2-cpt-val}.
    It only remains to show that the decision weights are computed correctly.
    For this, we first phrase the gain rank $\gainrank$ in terms of the probability measure in the MC.
    \begin{align*}
        \gainrank(\vec{x},i) &\eqdef \sum_{m = i+1}^k p_m \tag{Definition of gain rank}\\
        &= \sum_{m = i+1}^k \probability_{\MDP}^{\strat}[\obj(\infinitepath)=o_m] \tag{By \Cref{lem:3-defs-make-sense}}\\
        &= \probability_{\MDP}^{\strat}[\obj(\infinitepath)>o_i] \tag{By definition of the ordering of $\vec{o}$}
    \end{align*}
    By using the same argument, we can also get the analogue for the probability to get something equal or better: $
    p_i + \gainrank(\vec{x},i) = \probability_{\MDP}^{\strat}[\obj(\infinitepath)\geq o_i].$
    Using this, we now show that the decision weights for positive outcomes are defined as in \Cref{eq:2-cpt-val}:
    
    \begin{align*}
        \decweight_i(\vec{x}) &\eqdef \weightP(p_i + \gainrank(\vec{x},i)) - \weightP(\gainrank(\vec{x},i)) \tag{Definition of decision weight for positive outcomes}\\
        &=  \left( \weightP\left( \probability_{\MDP}^{\strat}(\obj(\infinitepath) \geq o) \right) - \weightP\left( \probability_{\MDP}^{\strat}(\obj(\infinitepath) > o) \right)\right) \tag{By the above analysis of gain rank}
    \end{align*}

    The proof for negative outcomes follows by an analogous analysis of the loss rank $\lossrank$.
    Thus, we conclude that \Cref{eq:2-cpt-val} and \Cref{eq:3-cpt-value} coincide, i.e.\ they describe the same value.
\end{proof}

\subsection{Reduction to Stopping Markov Decision Processes}\label{app:3-stopping-details}

For the formal proof, we first recall the definition of MEC Quotient MDP, building on~\cite[Alg. 3.3]{de1997formal}. Apart from adapting the notation, one key difference is that we introduce the $\stay$ action and the dedicated sink state~$\sink$ (inspired by~\cite{AshokCDKM17}). These additions are necessary, as we deal with weighted reachability where in the presence of negative rewards it can be advantageous to stay and obtain a reward of 0; in contrast, for the reachability objective considered in~\cite{de1997formal}, 0 is the worst value, so the $\stay$ action is omitted.

We assume w.l.o.g. that every action is unique, i.e.\ for all $s,s'\in\states$ we have $\act(s)\cap\act(s')=\emptyset$.

\begin{definition}[MEC Quotient MDP]\label{def:3-mec-quotient}
	Let $\MDP = (\states, \initstate, \act, \trans)$ be an MDP. 
	Let $\MECs_\MDP(\states\setminus\targets) = \{(T_1,A_1),\ldots,(T_n,A_n)\}$ be the set of all non-target MECs .
	Further, define $\MECs_\states=\bigcup_{i=1}^nT_i$ as the set of all states contained in some non-target MEC. The \emph{MEC quotient of $\MDP$} is defined as the MDP $\QMDP=(\Qstates, \Qinit, \Qact, \Qtrans)$, where:
	\begin{itemize}
		\item $\Qstates=\states\backslash \MEC_\states\cup\{\shat_1,\ldots,\shat_n\} \cup \{\sink\},$
		\item $\Qact = \act \cup \{\stay\}$
		\item if for some $T_i$, we have $\initstate\in T_i$, then $\Qinit=\shat_i$, otherwise $\Qinit=\initstate$
		\item for original states $s\in \states\backslash \MEC_\states$:
		$\Qact(s)=\act(s)$
		
		for collapsed states, i.e. for $\shat_i$ with $1\leq i\leq n$: $\Qact(\shat_i)=\{a\in\act\;|\; s\in T_i \land a\in\act(s) \land a\notin A_i\} \cup \{\stay\}$
		
		For the sink state, we have $\Qact(\sink) = \stay$.
		
		\item for the transition function, we need to address transitions coming from new states or leading to new states.
		Denote by $\mathsf{orig}(s,a)$ the state in the original MDP that is represented by $s$; the additional action $a$ allows us to uniquely identify a state in a collapsed MEC, since we assumed that actions are unique.
		Formally, $\mathsf{orig}(s,a) \eqdef s$ if $s\in\states$, and otherwise, i.e.\ if $s=\shat_j$, we have $\mathsf{orig}(s,a) \eqdef t$, where $t$ is the state in $T_j$ with $a \in\act(t)$.
		For all states $s\in\Qstates$ and actions $a\in\act$:
		
		$$\Qtrans(s,a,t)=\begin{cases}
			\sum_{s'\in T_j}\trans(\mathsf{orig}(s,a),a,s') & \text{if } t=\shat_j\\
			\trans(\mathsf{orig}(s,a),a,t)&\text{otherwise} %, i.e. } t\in\states\backslash\MEC_\states
		\end{cases}$$
		
		For all collapsed states, i.e.\ $\shat_i$ with $1\leq i\leq n$, the transition for the $\stay$ action is:
		$\Qtrans(\shat_i,\stay,\sink) = 1$.
		Similarly, we have $\Qtrans(\sink,\stay,\sink) = 1$.
		
	\end{itemize}
\end{definition}

Using this, we can now prove \cref{lem:3-stopping}.

\stopping*

\begin{proof}
	Let $\QMDP$ be the MEC Quotient of $\MDP$, as defined in \Cref{def:3-mec-quotient}.
	The MECs can be computed in polynomial time, and iterating over them once to replace each of them takes linear time.
	We now prove that $\QMDP$ is stopping: There are no MECs in $\Qstates\setminus (\targets \cup \{\sink\}$, i.e.\ in the part of the state space other than the absorbing targets and sink.
	As a MEC is reached almost surely~\cite[Thm. 3.2]{de1997formal}, we have that for every strategy $\Qstrat$ in $\QMDP$ it holds that 
	$\probability_{\QMDP}^{\Qstrat}[\reach \targets \cup \{\sink\}] = 1$.
	
	It remains to show that $\CPTval(\MDP) = \CPTval(\QMDP)$.
	To this end, we show that for every strategy $\strat$ in the original MDP, we can construct a strategy $\Qstrat$ in the quotient MDP such that the prospect of the induced MCs is the same, i.e.\ $\prospect(\MDP^\strat,\obj) = \prospect(\QMDP^\Qstrat,\obj)$, and vice versa.
	(Note that, as usual, we have fixed a weighted reachability objective $\obj$.)
	From this, it follows by the definition of CPT-value in \Cref{eq:3-cpt-value} that the CPT-value of both MDPs is the same.
	Note that, by the same argument and \Cref{eq:3-eu-value}, the expected value in both MDPs also coincides.
	
	Let $(T_i,A_i)$ be an arbitrary non-target MEC of $\MDP$, and let $\infinitepath s$ be the history so far, where the last state in the history $s$ is in the EC $T_i$.
	In the following, we use $\probability_{\MDP,s}^{\strat(\infinitepath)}$ to denote the probability measure in the MDP with $s$ as the initial state, where the strategy assumes that the history $\infinitepath$ is already given (and analogously for $\QMDP$ and $\Qstrat$).
	We now distinguish two cases:
	\begin{itemize}
		\item $\probability_{\MDP,s}^{\strat(\infinitepath)} [\reach \states\setminus T_i] = 0$, i.e.\ the play remains in $T_i$ forever. In this case, $\Qstrat(\infinitepath \shat_i)$ surely picks the action $\stay$.
		Then, we have that in the original MDP
		$\probability_{\MDP,s}^{\strat(\infinitepath)} [\reach \obtainset(\MDP^\strat(\infinitepath,\obj,0))] = 1$ and in the quotient MDP
		$\probability_{\QMDP,s}^{\Qstrat(\infinitepath)} [\reach \obtainset(\QMDP^\Qstrat(\infinitepath,\obj,0))] =1$.
		Intuitively, we have shown that we can mimic the staying behaviour.
		\item $\probability_{\MDP,s}^{\strat(\infinitepath)} [\reach \states\setminus T]_i > 0$, i.e.\ there is a chance to leave the MEC.
		Intuitively, in this case we need to compute the distribution over all leaving actions in the original MDP so that we can mimic the leaving behaviour in the quotient MDP.
		Formally, this is cumbersome to write down.
		
		Consider the countable induced MC $\MDP^\strat$ where paths in the MDP $\MDP$ are states (see~\cite[Def. 10.92]{BK08} for details on this standard notion), and let the initial state be $\infinitepath s$.
		Modify this MC as follows: For every state $\infinitepath s \infinitepath'$, make it absorbing if there exists a state $t\in\infinitepath'$ such that $t\notin T_i$.
		Every of these new absorbing states has the form is labelled $\infinitepath s \infinitepath' a t$, where $a$ is the action that was picked before $t$ was reached. 
		Thus, by computing the limiting distribution in this modified MC, we can infer the probability distribution over the leaving actions.
		This distribution is exactly the probability distribution over actions that we use as strategy $\Qstrat$ in the quotient MDP.
		Thus, for every target state, the probability to reach it in $\MDP$ and $\QMDP$ is equal.
	\end{itemize}
	
	We comment on a further technicality: Reaching different states in the MEC can lead to different behaviour. 
	For example, assume that after seeing the history $\infinitepath$ there is an equal chance of reaching $s_1$ and $s_2$, both of which are in $T_i$.
	The strategy might recommend staying in $s_1$, but leaving in $s_2$.
	In the quotient MDP, we cannot distinguish these cases, as the history only contains $\shat_i$, the representative of the whole MEC.
	
	This is not a problem for two reasons.
	Firstly, our overall goal is to find the optimum strategy.
	Since all states in a MEC can reach each other with probability 1, all states in a MEC have the same value. 
	Thus, an optimal strategy need not distinguish between states in the MEC.
	
	Secondly, we can modify the definition of quotient MDP slightly to account for this fact; this is necessary to prove our current sub-goal, i.e.\ that we can mimic the behaviour of every strategy such that $\prospect(\MDP^\strat,\obj) = \prospect(\QMDP^\Qstrat,\obj)$; intuitively, we also want to be able to mimic \enquote{stupid} strategies that differentiate between different states in a MEC.
	To this end, we keep every original MEC state in the quotient, but only give it one action that surely transitions to the representative $\shat_i$.
	This representative is unchanged. 
	(The transition function then becomes simpler, because we do not have to distinguish transitions leading into a MEC and sum the probability of all MEC states, but we can just go to the original state.)
	In this quotient MDP, the history now contains the original state which allows us to also mimic behaviour that differs depending on the entry state.
	
	\medskip
	Overall, we have shown that we can indeed construct a strategy $\Qstrat$ such that for every strategy $\strat$ in the original MDP, we can mimic its behaviour in the quotient MDP using $\Qstrat$.
	The proof for the other direction is mostly analogous and even slightly simpler.
	The $\stay$ action can be replaced with a strategy that chooses an arbitrary subset of actions in $A_i$.
	A distribution over leaving actions is achieved by picking an appropriate memoryless randomized strategy inside the MEC.
	Formally, this can be achieved as follows:
	Every leaving action can be considered a target set of a multi-objective reachability query. 
	Since every state in the MEC can reach every other state almost surely (by definition of MEC), the Pareto frontier of this query contains all probability distributions over leaving actions.
	Thus, using~\cite[Cor. 3.5]{EKVY08}, we can obtain a memoryless randomized strategy that achieves exactly the distribution over leaving actions that is prescribed by $\Qstrat$.
	
	Thus, we get that the set of all induced prospects in both MDPs are equal, and can conclude that the CPT-value is equal.
	
\end{proof}
This lemma allows us to assume that w.l.o.g. the MDP we are given is stopping because given an arbitrary MDP, we can transform it to a stopping MDP with the same CPT-value and then compute the CPT-value in the stopping MDP.

\begin{figure}[t]
	\centering
	 \begin{tikzpicture}[scale=1.4]
    \node[draw,circle, minimum size=0.5cm] (s0) at (0,0){$s$};

    \node[draw,circle, minimum size=0.5cm] (t) at (1.5,0){$t$};
    
    \node (tc) at (2,0){$-5$};

    \draw[->,thick] (-.75,0) -- (s0);
    \draw[->,thick] (s0) --  node[above] {b} (t);
     \draw[->,thick] (s0) edge[loop above] node[above] {a} (s0);
     \draw[->,thick] (t) edge[loop above] node[above] {a} (t);

\end{tikzpicture}
	\caption{A simple example of a non-stopping MDP. Transition probabilities are omitted, as they are equal to 1 (repeat of \cref{fig:3-stopping} for convenience).
	}
	\label{fig:app-3-stopping}
\end{figure}
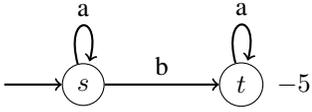

\begin{example}
	We briefly exemplify the necessity of the assumption that the MDP is stopping for the correctness of the multi-objective reachability query defined in \cref{eq:3-mo-query}:
	For the MDP in \Cref{fig:app-3-stopping}, there are two outcomes, 0 and $-5$. 
	However, the $\obtainset$ of outcome 0 is empty.
	Hence, there is only a single point on the Pareto frontier of the corresponding multi-objective query, namely the one that gives probability 1 to outcome $-5$.
	However, there is the additional prospect [-5:0, 0:1] (with better CPT-value); it is induced by playing the staying action.
	This prospect is not a point on the Pareto frontier.
	Thus, without the stopping assumption, the Pareto frontier does not contain all possible prospects. \qee
\end{example}

%\cptvalpf*

\subsection{Lower Bound on Strategy Complexity}
\label{app:examples-rand}

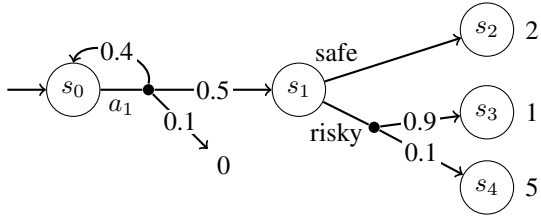
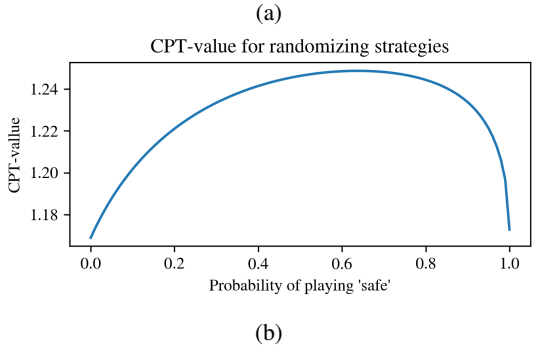
\begin{figure}[t]
    \centering
    \begin{subfigure}{0.4 \textwidth}
        \begin{tikzpicture}
    \node[] (s0) at (0,0) {};
    \node[draw, circle, minimum size=0.5cm, right of=s0] (s) {$s_0$};
    \node[draw, circle, minimum size=0.5cm, right of=s, xshift=2cm] (t) {$s_1$};
    \node[right of=s, fill=black,inner sep=1.5pt, circle] (bd) {};
    \node[draw, circle, right of=t, minimum size=0.5cm, xshift=1.5cm, yshift=0.8cm] (u) {$s_2$};
    \node[draw, circle, right of=t, minimum size=0.5cm, xshift=1.5cm, yshift=-0.3cm] (v) {$s_3$};
    \node[draw, circle, right of=t, minimum size=0.5cm, xshift=1.5cm, yshift=-1.3cm] (w) {$s_4$};
    \node[right of=t, fill=black,inner sep=1.5pt, circle, yshift=-0.5cm] (bd2) {};
    \node[right of=s,yshift=-1cm,xshift=1cm] (zero) {0};
    \node[right of=u, xshift=-0.4cm] () {2};
    \node[right of=v, xshift=-0.4cm] () {1};
    \node[right of=w, xshift=-0.4cm] () {5};

    \draw[->,thick] (s0) -- (s);
    \draw[-,thick] (s) edge[below] node{$a_1$} (bd);
    \draw[->, thick] (bd) edge[bend right=90] node[fill=white, inner sep=1pt]{0.4} (s);
    \draw[->,thick] (bd) edge[] node[fill=white, inner sep=1pt]{0.1} (zero);
    \draw[->,thick] (bd) edge[] node[fill=white, inner sep=1pt]{0.5} (t);
    \draw[-,thick] (t) edge[] (bd2);
    \draw[->, thick] (t) edge[]  (u);
    \node[right of=t,xshift=-0.5cm,yshift=0.5cm] () {safe};
    \node[right of=t,xshift=-0.5cm,yshift=-0.6cm] () {risky};
    \draw[->, thick] (bd2) edge[] node[fill=white, inner sep=1pt]{0.9} (v);
    \draw[->, thick] (bd2) edge[] node[fill=white, inner sep=1pt]{0.1} (w);
    
\end{tikzpicture}
        \caption{}
        \label{fig:app-mdp-random}
    \end{subfigure}
	\hfill
    \begin{subfigure}{0.4 \textwidth}
        \includegraphics[width=\linewidth]{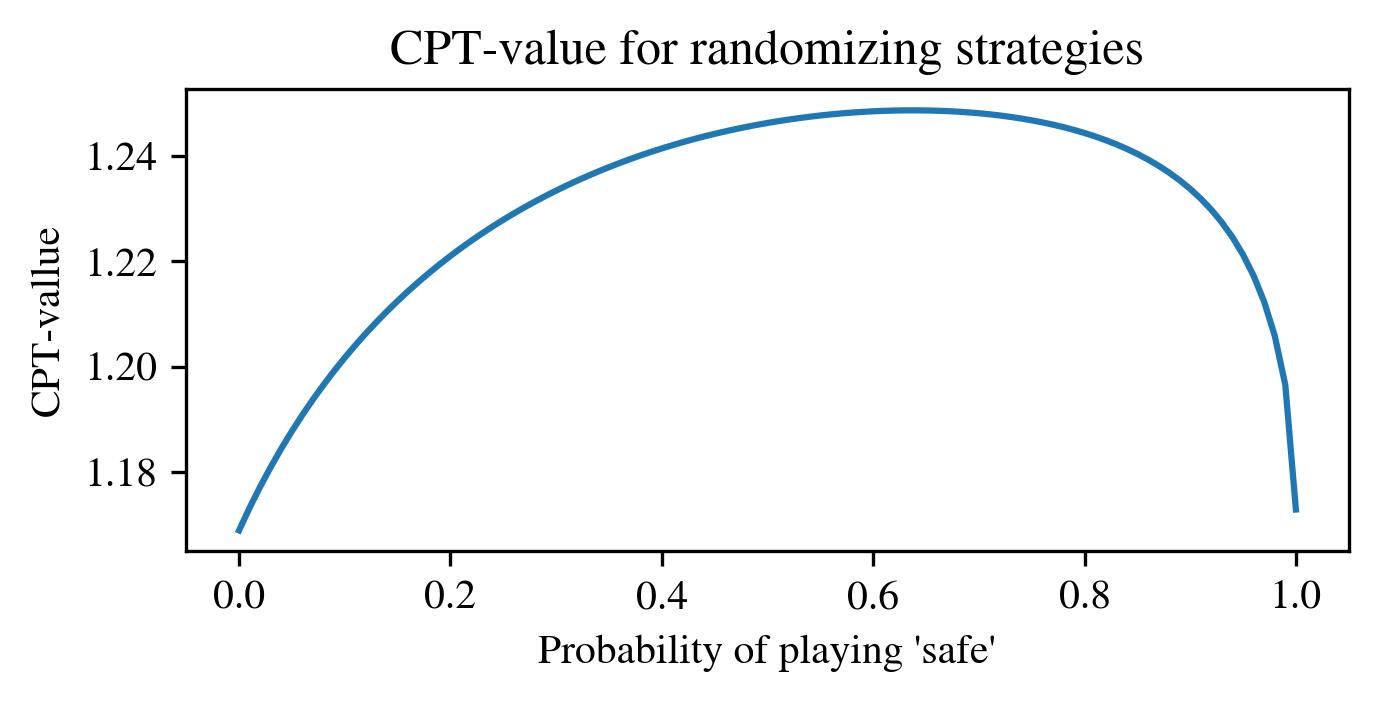}
        \caption{}
        \label{fig:app-mdp-plot}
    \end{subfigure}
    \caption{MDP for demonstrating that an optimal strategy can be randomizing. (a) The MDP. (b) CPT-value of the MDP for different strategies, computed with the standard functions of \cref{app:2-u-and-w}.}
\end{figure}
Consider the example MDP in \cref{fig:app-mdp-random}.
In $s_0$, there is no choice to take, as there is only one action available.
This will lead back to the initial state with a probability of 0.4, to a (not drawn) sink-state with a probability of 0.1 and it will move to state $s_1$ with a probability of 0.5.
In state $s_1$, we are faced with a decision: take the \emph{safe} option and gain 2 or choose the \emph{risky} option where we gain 1 with a probability of 0.9 and 5 otherwise.

There are two deterministic strategies available: playing \emph{safe} or playing \emph{risky}.
But we can see in \cref{fig:app-mdp-plot} that the optimal choice is neither of the two deterministic strategies, but it is achieved (approximately) when we play \emph{safe} with probability 0.7, and \emph{risky} with 0.3.

	\section{Additional Details for Section~\ref{sec:4-title}}\label{app:4-title}

\subsection{Non Pseudo-Convexity of $\cptfun$}\label{app:4-pseudoconvex}
\begin{lemma}[Non-Pseudo-Convexity]
	\label{lem:cpt-not-pseudoconvex}
	The CPT function is not pseudo-convex (and thus not convex).
\end{lemma}

\begin{proof}
	We denote by $\nabla f$ the gradient of $f=(\frac{\partial  f}{x_1},\ldots,\frac{\partial  f}{x_n})$.
	A function $f:X\rightarrow \Reals$ is pseudo-convex (see, e.g.,~\cite{pseudo-convex}) if the following holds for all $x,y\in X$: $$\nabla f(x)^T(y-x)\geq 0\Rightarrow f(y)\geq f(x)$$
	
	To show that $\cptfun$ is not generally pseudo-convex, we provide a counterexample.
	Assume, we have two prospects $\prosp^1$ and $\prosp^2$, with the outcomes $\outcomes=\{0,1,2\}$
	We have \begin{align*}
		\cptfun(\prosp^1)=&\util(2) \cdot w(p^1_1) \\
		&+\util(1) \cdot (w(p^1_1+p^1_2)-w(p^1_1))\\
		&+ \util(0) \cdot (1-w(p^1_1+p^1_2))
	\end{align*}
	and equivalently
	\begin{align*}
		\cptfun(\prosp^2)=&\util(2) \cdot w(p^2_1) \\
		&+ \util(1) \cdot (w(p^2_1+p^2_2)-w(p^2_1)) \\
		&+ \util(0) \cdot (1-w(p^2_1+p^2_2))
	\end{align*}
	
	We have that 
	\begin{align*}
		\nabla \cptfun(\prosp^1)^T=[\util(2) \cdot w'(p^1_1) + \util(1) \cdot (w'(p^1_1+p^1_2)-w'(p^1_1)),\\ \util(1)\cdot w'(p^1_1+p^1_2)]
	\end{align*}
	where $w'(x)$ is the derivative of $w$ (which we omit to write out for readability).
	
	If we assume the prospects to have the following numbers, $\prosp^1=[2:0.01, 1:0.31, 0:0.68]$, $\prosp^2=[2:0.01,1:0.01, 0:0.98]$, and using $\util(x)=x$, and $$w(x)=\frac{x^{2.25}}{(x^{2.25}+(1-x)^{2.25})^{1/{2.25}}}$$
	we get that $$\nabla \cptfun(\prosp_1)^T(\prosp_2-\prosp1)\;\approx\;0.245>0$$
	but $$\cptfun(\prosp_2)-\cptfun(\prosp_1)\approx-0.105<0$$
	
	Therefore, we found a violation of the pseudo-convexity condition.
	Since it is known that pseudo-convexity implies convexity~\cite{pseudo-convex}, this also proves (what by inspection of $w$ and $\util$ we already assumed) that in general $\cptfun$ is not convex.
\end{proof}

\subsection{Non-monotonicity of the CPT-function}\label{app:4-cpt-increasing}

The CPT-function is not monotonic with respect to the probability vector in the prospect. Intuitively, this is the case because increasing the probability of a negative outcome decreases the CPT-function.
We formalize this intuition in \Cref{lem:4-cpt-increasing}.
Further, we show that it is in fact a \emph{difference-of-monotonic} function (\cref{lem:4-cpt-dom}); this notion is relevant in non-convex optimization, cf.~\cite[11.1.2]{tuy1998convex}. We do not utilize this fact, but include it for completeness.

\begin{restatable}{lemma}{cptinc}
	\label{lem:4-cpt-increasing}
	Assume that $\util$, $\pweight$, and $\mweight$ are increasing functions.
	The CPT function $\cptfun((\outcomevector, \probabilities))$ is monotonically increasing in $\probabilities$ for only positive outcomes and monotonically decreasing for only negative outcomes. Formally, for two vectors $\probabilities\geq\probabilities'$, we have
	\begin{enumerate}
		\item if $\outcomevector>\vec{0}$ then $\cptfun((\outcomevector,\probabilities))\geq\cptfun((\outcomevector,\probabilities'))$
		\item if $\outcomevector<\vec{0}$ then $\cptfun((\outcomevector,\probabilities))\leq\cptfun((\outcomevector,\probabilities'))$
	\end{enumerate}
\end{restatable}
\begin{proof}
	The proof is straightforward by rewriting the CPT-function. 
	
	The CPT function satisfies stochastic dominance \cite{wakker2010prospect}, but we want to show monotonicity nevertheless since the definition in \cite{wakker2010prospect} lacks a precise mathematical characterization.
	As usual, w.l.o.g. let $\outcomevector$ be ordered increasingly.
	
	\textbf{Property (1)}
	%Let us first look at property (1).
	Recall that for $\vec{o}>\vec{0}$, we have 
	$$\cptfun((\vec{o},\vec{p}))=\sum_{i=1}^k\util(o_i)\cdot\left(\pweight\left(\sum_{j=i}^kp_j\right)-\pweight\left(\sum_{j=i+1}^kp_j\right)\right)$$

	We regroup this sum, such that each term only looks at one sum of probabilities $\sum_{j=i}^kp_j$:
	\begin{align*} &\cptfun((\vec{o},\vec{p}))\\
		&=\left(\sum_{i=2}^{k}\util(o_i)\cdot\pweight\left(\sum_{j=i}^kp_j\right)-\util(o_{i-1})\cdot\pweight\left(\sum_{j=i}^kp_j\right)\right)\\
		&\quad +\util(o_1)\cdot\pweight\left(\sum_{j=1}^kp_j\right) - \util(o_k)\cdot\pweight(0)\\
		&=\left(\sum_{i=2}^{k}(\util(o_i)-\util(o_{i-1}))\cdot\pweight\left(\sum_{j=i}^kp_j\right)\right)\\
		&\quad +\util(o_1)\cdot\pweight\left(\sum_{j=1}^kp_j\right)
	\end{align*}
	We know from the increasing ordering of $o$ and because $\util$ is an increasing function that $\util(o_i)\geq\util(o_{i-1})$.
	Since also $\pweight$ is an increasing function, and since we have the assumption $\vec{p}\geq\vec{p}'$, we have for each term 
	\begin{align*}
		&(\util(o_i)-\util(o_{i-1}))\cdot\pweight\left(\sum_{j=i}^kp_j\right)\geq\\
		&(\util(o_i)-\util(o_{i-1}))\cdot\pweight\left(\sum_{j=i}^kp_j'\right)
	\end{align*}
	and clearly $$\util(o_1)\cdot\pweight\left(\sum_{i=1}^kp_i\right)\geq\util(o_1)\cdot\pweight\left(\sum_{i=1}^kp_i'\right).$$ 
	
	Thus, we have $\cptfun((\vec{o},\vec{p}))\geq\cptfun((\vec{o},\vec{p}'))$.\\
	
	\textbf{Property (2)}
	The proof works analogously. 
	We have to consider that $\util$ is now negative, and we use the weighting function $\mweight$ instead of $\pweight$.
	The regrouping works equivalently, and the arguments of an increasing $\util$ and increasing $\mweight$ still hold.
\end{proof}

\begin{corollary}
	\label{lem:4-cpt-dom}
	The CPT-function $\cptfun$ is a difference-of-monotonic function.
\end{corollary}
\begin{proof}
	The CPT-function is a sum over all outcomes.
	We can also write it as the difference between the following two monotonic functions:
	firstly the CPT-function on only the positive outcomes, and secondly $(-1)$ times the CPT-function on only the negative outcomes.
	By \cref{lem:4-cpt-increasing}, both of these functions are monotonically increasing (the latter since the CPT-function on negative outcomes is decreasing, and we multiply it with $(-1)$).
\end{proof}

\subsection{Lipschitz-continuity of CPT-function}\label{app:4-cpt-lipschitz}

\cptLc*
\begin{proof}
	The following chain of equations proves our goal
	\begin{align*}
		\abs{\cptfun&((\outcomevector,\probabilities))-\cptfun((\outcomevector,\probvectorQ))} \\
		&= \abs{\sum_{i=1}^k\util(o_i)\decweight(p_i)-\sum_{i=1}^k\util(o_i)\decweight(q_i)}
		\tag{Unfolding definition of $\cptfun$}\\
		&\leq \sum_{i=1}^k \abs{\util(o_i)(\decweight(p_i)-\decweight(q_i))}
		 \tag{Reordering and triangle inequality}\\
		&\leq \util^* \cdot \sum_{i=1}^k \abs{(\decweight(p_i)-\decweight(q_i)}
		\tag{Definition of $\util^*$}\\
		&\leq \util^* \cdot \varepsilon \cdot \max(\lipWp,\lipWm) \cdot (2k^2+k) \tag{See below}\\
		&= \varepsilon \cdot \lipCPT \tag{Definition of $\lipCPT$}
	\end{align*}

    It remains to prove:
	\[
	\sum_{i=1}^k \abs{(\decweight(p_i)-\decweight(q_i)} \leq \varepsilon \cdot \max(L_{\pweight},L_{\mweight}) \cdot (2k^2+k).
	\]
	
	We first analyse the summand, i.e.\ the difference of decision weights. 
	The definition of decision weight $\decweight$ distinguishes whether an outcome is 0, positive, negative.
	If it is 0, then its decision weight is also 0 independent of the associated probability being $p_i$ or $q_i$, and thus it does not contribute to the difference.
	We now focus on the case of a positive outcome, i.e.\ $o_i>0$:
	
	\begin{align*}
		&\left|\decweight(p_i)-\decweight(q_i)\right|\\
		&=|\left(\weightP(p_i + \gainrank(\probvector,i)) - \weightP(\gainrank(\probvector,i))\right)\\
		&\quad\quad - \left(\weightP(q_i + \gainrank(\probvectorQ,i)) - \weightP(\gainrank(\probvectorQ,i))\right)|
		\tag{Unfolding definition of $\decweight$}\\
		&=\Bigg|\left(\weightP(\sum_{m = i}^k p_m) - \weightP(\sum_{m = i+1}^k p_m)\right)\\
		&\quad\quad - 
		\left(\weightP(\sum_{m = i}^k q_m) - \weightP(\sum_{m = i+1}^k q_m)\right) \Bigg|
		\tag{Unfolding definition of $\gainrank$}\\
		&\leq \left|\weightP(\sum_{m = i}^k p_m) - \weightP(\sum_{m = i}^k q_m)\right|\\
		&\quad\quad + 
		\left|\weightP(\sum_{m = i+1}^k q_m)) - \weightP(\sum_{m = i+1}^k p_m)\right|
		\tag{Reordering and triangle inequality}
	\end{align*}
	Now, we use the fact that $\weightP$ is Lipschitz-continuous on [0,1] (as it is monotonic and bounded by [0,1] by assumption).
	Thus, we need to bound the difference between its arguments.
	We have
	\begin{align*}
		\left|\sum_{m = i}^k p_m - \sum_{m = i}^k q_m\right|\\
		= \left|\sum_{m = i}^k p_m - q_m \right|\tag{Reordering}\\
		\leq \sum_{m = i}^k \max_{j \in [i,k]} \left|p_j - q_j\right| \tag{Using the maximum component-wise difference and triangle inequality}\\
		\leq \sum_{m = i}^k \varepsilon \tag{As $\abs{\probabilities - \probvectorQ} \leq \varepsilon$}\\
		= (k-i) \cdot \varepsilon
	\end{align*}
	By the analogous argument, we get that 
	\[\left|\sum_{m = i+1}^k q_m - \sum_{m = i+1}^k p_m\right| \leq (k-i+1) \cdot \varepsilon.\]
    We denote the Lipschitz-constant of $\weightP$ by $\lipWp$.
	Utilizing these facts and the Lipschitz-continuity of $\weightP$, we obtain:
	\begin{align*}
		&\left|\decweight(p_i)-\decweight(q_i)\right|\\
		&\leq \left(\lipWp \cdot (k-i) \cdot \varepsilon\right) + \left(\lipWp \cdot (k-i+1) \cdot \varepsilon\right)\\
		&= \varepsilon \cdot \lipWp \cdot (2k-2i+1)
	\end{align*}

	We repeat the argument for negative outcomes. The only difference is in the borders of the summations, now running from 1 to $i$ or from 1 to $i-1$. Thus, the last factor becomes $i-1 + i-2$, which is equal to $2i - 3$.
	Overall, for $o_i < 0$, we have
	\[\left|\decweight(p_i)-\decweight(q_i)\right| \leq \varepsilon \cdot \lipWm \cdot (2i-3)\]
	
	We unify the expressions for positive and negative outcomes. 
	For this, we take the larger of the two Lipschitz-constants $\max(\lipWp,\lipWm)$, and upper bound the last factor by $2k+1$ (which is certainly larger than $2i-3$, as $i\leq k$).
	Thus, we have 
	\[
		\left|\decweight(p_i)-\decweight(q_i)\right| \leq \varepsilon \cdot \max(\lipWp,\lipWm) \cdot (2k+1)
	\]
	Now we have analysed the summand, and in the process eliminated the dependence on $i$.
	Applying our knowledge, we conclude with the desired claim:
	\begin{align*}
	\sum_{i=1}^k \abs{(\decweight(p_i)-\decweight(q_i)} &\leq k \cdot \left(\varepsilon \cdot \max(\lipWp,\lipWm) \cdot (2k+1)\right)\\
	&= \varepsilon \cdot \max(\lipWp,\lipWm) \cdot (2k^2+k)
	\end{align*}

\end{proof}

	\section{Cumulative Prospect Theory for Other Objectives}\label{app:5-mp}

\subsection{CPT-value for Mean Payoff Objectives}\label{app:mp-prelim}

To define the CPT-value of a mean payoff (MP) objective, we employ the characterization of CPT-value through \cref{eq:3-cpt-value}.
To do this, we need to define the $\obtainset$ on an MC with an MP-objective.
For a BSCC $C$ in an MC, we can compute its mean payoff by solving a system of linear equations, see e.g.~\cite[Chap. 8]{puterman}.
Let $\meanpay(C,\objMP)$ be the mean payoff of a BSCC.

\begin{align}
    \obtainset(\MC,\objMP,o_i) &\eqdef \nonumber\\
    &\{C \in \BSCCs(\MC) \mid \meanpay(C,\objMP) = o_i\} 
    \label{eq:MP-obtainset}
\end{align}
Note that we do not need a case distinction to treat the 0-outcome separately.
Given this definition of $\obtainset$, the definition of prospect induced by the MC is the same for both WR- and MP-objective.

Given this definition of CPT-value, we repeat the key claim.
\mpreduction*
The next sections provide the formal proof.

\subsection{Reduction of Mean Payoff to Weighted Reachability}\label{app:mp-reduction}

The MEC Quotient MDP describes an MDP where all MECs are collapsed into single states.
We modify the MEC quotient MDP as described in \cref{def:3-mec-quotient} to include the reward of a BSCC (whereas previously BSCCs always led to a sink with reward 0).

The Weighted MEC Quotient (similar to the one presented in \cite{AshokCDKM17}) adds to the MEC Quotient MDP a possibility to compute rewards on the collapsed MECs.
Intuitively, it assumes that we already know the outcome of a MEC. 
We add another action to each collapsed state ($\stay$) that represents the choice of staying indefinitely in the MEC and thus gaining the outcome. Alternatively, it is still possible to leave the MEC (if this was already possible before).
\begin{definition}[Weighted MEC Quotient]
\label{def:weighted-mec-quotient}
    Let $\MEC_{\Qstates}=\{\shat_1,\ldots,\shat_n\}$ and $f\colon \MEC_{\Qstates}\rightarrow \Rationals$ be a function assigning a value to every collapsed state. 
    We define the \emph{weighted MEC quotient of $\MDP$ and $f$} as the MDP $\MDP^f=(\states^f,\initstate^f,\act^f,\trans^f)$, where
    \begin{itemize}
        \item $\states^f=\Qstates\backslash\{\sink\}\cup\{\shat_1^f,\ldots,\shat_n^f\}$ for each collapsed state $\shat_i$
        \item $\initstate^f=\Qinit$
        \item $\act^f=\act$
        \item $\trans^f$ is defined as 
        $$\forall \shat\in\Qstates,\widehat{a}\in\Qact\backslash\{\stay\}.\;\trans^f(\shat,\widehat{a})=\Qtrans(\shat,\widehat{a})$$
        $$\forall\shat_i\in\MEC_{\Qstates}.\;\trans^f(\shat_i,\stay)=(\shat_i^f \rightarrow 1)$$
    \end{itemize}
    We now define the reward function for weighted reachability on the target states $\{\shat_1^f,\ldots,\shat_n^f\}$ and the reward to be $r(\shat_i^f)=f(\shat_i)$.
\end{definition}

We show how we can transform an MDP with a mean payoff objective into an MDP with weighted reachability objective, such that the optimal prospects for each of them are the same. 
To this end, we introduce the transformation into a \emph{MEC Quotient MDP} and consequently, a \emph{Weighted MEC Quotient MDP}. 
In \cref{th:mp-wr-mec}, we show that the optimal prospects are the same.
Finally, we conclude in \cref{cor:cpt-mp-is-cpt-wr} that the CPT-value of the MDP with mean-payoff objective is the same as the transformed MDP with weighted reachability objective.

\begin{lemma}\label{th:mp-wr-mec}
With the transformation from \cref{def:weighted-mec-quotient}, we have that 
$$\sup_\strat\prospect(\MDP^\strat,\objMP)=\sup_\strat\prospect({\MDP^f}^\strat, \obj^{WR})$$
for $f(\hat{s_i})=\max_{\strat}r^\strat(\hat{s_i})$, where each MEC gets the highest possible outcome possible under the reward function for the objective.
\end{lemma}
\begin{proof}
We first demonstrate how one can intuitively transform the mean payoff MDP to a weighted reachability MDP, but the resulting MDP might be exponentially bigger.
Then we show that we can restrict ourselves to a much smaller subset, making the transformation polynomial.

We can find a function $f$ for the transformation by looking at a strategy $\strat$ for the original MDP $\MDP$. 
For a fixed strategy $\strat$, each MEC in the MDP will be either transient (i.e.\ not a BSCC in the induced MC, and as such irrelevant for the mean-payoff, because we only care about the infinite tail) or a BSCC and thus, the only relevant part of the MDP for the mean payoff objective.
Therefore, for each MEC $C_i=(T_i, A_i)$, we can compute its ``reward''. 
$$r(C_i)=\begin{cases}
    0&\text{if the MEC is transient}\\
    r(C_i^\strat)&\text{else}
\end{cases}$$
Note that we can compute the mean payoff in an MEC by computing its steady state distribution and directly computing the mean payoff from that (see e.g. \cite{DBLP:conf/tacas/Meggendorfer23}). 
Therefore, any strategy induces a function $f^\strat$ with which we can transform the MDP to a weighted MEC quotient MDP. 

For a fixed strategy $\strat$, and thus $\MDP$ being an MC $\MDP^\strat$, it is easy to see that $\prospect(\MDP^\strat,\objMP)=\prospect({\MDP^f}^\strat)$.
In the steady state distribution, any transient state has a probability of 0. 
Therefore, we only care about the BSCCs.
In the transformed MDP $\MDP^f$, all BSCCs are single states.
In the original MDP, they might be actual BSCCs containing several states.
However, we transformed each such BSCC into one single state in the transformed MDP, with the exact mean payoff value as a reward.

For an easier understanding, let us look at a weighted MEC quotient MDP, where we not only look at one function $f$ but at finitely many functions $f_1,\ldots f_r$. 
Intuitively, each of them can be derived from a strategy $\strat_i$.
Assume, we look at the MDP $\MDP^{f_1,\ldots,f_r}$, where we add states $\{\hat{s_1}^{f_i},\ldots,\hat{s_n}^{f_i}\}$ for all functions $f_i$, where $f_i(\hat{s_j})=r^{\strat_i}(\hat{s_j})$. 
For all possible rewards of an MEC, we add a state to the MDP that represents this reward.
First, prove that if we add a stay action for every BSCC that can result from the MEC, then we have exactly the same set of prospects.
Then, we can easily show that \begin{align*}
\{\prospect(&\MDP^\strat,\objMP)\;|\;\strat\in\Sigma\}=\\
&\{\prospect(\MDP^{f_i},\objWR)\;|\;\forall o_i\;\text{MD strategies}\}\end{align*}
Note that there are only finitely many memoryless deterministic strategies for a finite MDP.

We show equality by proving inclusion in both directions.

(i) Suppose we are given a prospect in the original mean payoff MDP.
This means there must be a strategy $\sigma_j$ under which we can achieve a set of outcomes with certain probabilities.
Then, each outcome must be achieved in a BSCC, and since we added a specific state with this outcome to the transformed MDP $\MDP^{f_j}$, we also have the same outcome there.
On the other hand, the probability of such an outcome is only determined by the probability of reaching its respective BSCC.
Since we don't change transient states in the transformed MDP, the probabilities are still the same.

(ii) If we are given a prospect in the transformed MDP, each outcome is received in a target state.
If this target state was already a state in the original MDP, the outcome must still be the same in the original.
If the state is a collapsed state, it comes from an MEC in the original MDP, for which there must exist a strategy such that this BSCC can achieve the exact outcome as a mean payoff. 
Thus, there must be a strategy that will reach this BSCC and returns the outcome as mean payoff.\\

However, this new MDP is now potentially exponentially bigger than before.
But we know that we are not interested in all additional states $\hat{s_i}^f$.
If $r(\hat{s_i}^{f_k})<r(\hat{s_i}^{f_l})$, we definitely prefer the latter and it will always give a greater prospect.
Thus, we can ignore all new stay actions with less than maximum reward.

Thus, overall, we add one stay action per MEC in the weighted MEC quotient, and certainly, the sup over strategies in the weighted reach MDP is the same as in the MP MDP.
\end{proof}
Note that the MDP $\MDP^f$ resulting from this transformation is immediately stopping.
Thus, we do not need to apply \cref{lem:3-stopping} when calculating the CPT-value of $\MDP^f$.

\mpreduction*
\begin{proof}
    This follows immediately from \cref{th:mp-wr-mec}.
    Note that it takes polynomial time to compute the MECs as well as the mean payoff of every MEC for the weighted MEC quotient.
\end{proof}

\subsection{Effect of the Number of Outcomes}\label{app:5-num-outcomes}

As an interesting side-observation, we discuss the effect of the number of outcomes on the applicability of decision theories in general, and CPT in particular.

\paragraph{Objectives with Two Outcomes.}\label{sec:num-outcomes}
Switching from expected utility to CPT only makes a difference for the resulting strategy if we have more than two possible outcomes. 
Intuitively, when there are only two outcomes with one being strictly better than the other, an optimal strategy will try to put as much probability on the better outcome as possible.
Thus, both CPT or expected utility will always yield a higher value for more probability mass in the higher outcome.
More generally, every function for evaluating prospects that is monotonic results in the same set of optimal strategies.
Thus, it does not make sense to investigate risk-aware MDPs with these kinds of objectives, including non-weighted reachability or $\omega$-regular objectives.

\paragraph{Objectives with Infinitely Many Outcomes.}
	We remark that our theory does not apply to objectives with infinitely many outcomes such as 
	total reward~\cite[Chap. 7]{puterman}, where the reward is obtained along the path.
	Since no set of states uniquely corresponds to an outcome, there is no $\obtainset$.
	Even if there was, the multi-objective reachability query would have infinitely many target sets.
	Already for a Markov chain, the computation of a prospect is non-trivial, as the following example shows:

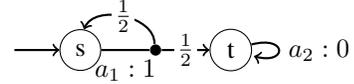
\begin{figure}
\centering
    \begin{tikzpicture}
    \node[] (s0) at (0,0) {};
    \node[draw, circle, minimum size=0.5cm, right of=s0] (s) {s};
    \node[draw, circle, minimum size=0.5cm, right of=s, xshift=1cm] (t) {t};
    \node[right of=s, fill=black,inner sep=1.5pt, circle] (bd) {};

    \draw[->,thick] (s0) -- (s);
    \draw[-,thick] (s) edge[below] node{$a_1:1$} (bd);
    \draw[->,thick] (bd) edge[bend right=80] node[fill=white,inner sep=1pt]{$\frac{1}{2}$} (s);
    \draw[->,thick] (bd) edge[] node[fill=white,inner sep=1pt]{$\frac{1}{2}$} (t);
    \draw[->,thick] (t) edge[loop right] node{$a_2:0$} (t);
    
\end{tikzpicture}
    \caption{Small example of Markov Chain with total reward}
    \label{fig:app-mc-total-reward}
\end{figure}
\begin{example}
    Let us consider the Markov chain in \cref{fig:app-mc-total-reward}.
    It has two states and one transition that returns to state $s$ with a probability of 0.5 and continues toward state $t$ with a probability of 0.5.
    On each iteration of action $a_1$, the total reward is increased by 1.
    Once we reach state $t$, we loop there and the total reward is no longer changed.
    The prospect of this MC is $\prosp=[1:\frac{1}{2}, 2:\frac{1}{4}, \ldots, n:\frac{1}{2}^n, \ldots]$ and it is infinite.
    Thus, the CPT value is given by
    $$\CPTval=\sum_{n=1}^\infty \util(n)\left(\pweight\left(\frac{1}{2}^{n-1}\right)-\pweight\left(\frac{1}{2}^n\right)\right)$$
    Even for such a small example, it is non-trivial to compute the actual value.
    \qee
\end{example}
    }
\end{document}